\documentclass[a4paper,12pt]{article}
\usepackage[utf8]{inputenc}

\usepackage [printonlyused, withpage]{acronym}
\usepackage {amsmath}
\usepackage {amssymb}
\usepackage {amsthm}
\usepackage [url=true, backend=biber, style=alphabetic, citestyle=alphabetic]{biblatex}
\usepackage {booktabs}
\usepackage {caption}
\usepackage {float}
\usepackage {geometry}
\usepackage {hyperref}
\usepackage {import}
\usepackage {mathtools}
\usepackage {pdfpages}
\usepackage {pgf}
\usepackage {subcaption}
\usepackage {upgreek}
\usepackage {url}

\makeatletter
\title{Zero Values of the TOV Equation}\let\Title\@title
\author{Jonas Pleyer}\let\Author\@author
\date{09.07.2021}\let\Date\@date
\makeatother

\numberwithin{equation}{subsection}

\newcommand{\Z}{\mathcal{Z}}
\newcommand{\F}{\mathcal{F}}
\newcommand{\Ent}{\mathcal{S}}
\newcommand{\U}{\mathcal{U}}
\newcommand{\R}{\mathbb{R}}
\newcommand{\e}{\mathit{e}}

\newcommand{\diff}{\mathop{}\!\mathrm{d}}

\DeclareUnicodeCharacter{2212}{-}

\newtheoremstyle{defstyle}
	{\topsep} 											
	{\topsep} 											
	{}													
	{}													
	{\bfseries}{}										
	{\newline}											
	{\thmname{#1}~\thmnumber{#2}\thmnote{\ -\ #3}}

\newtheoremstyle{normalstyle}
	{\topsep} 											
	{\topsep} 											
	{}													
	{}													
	{\bfseries}{}										
	{\newline}											
	{\thmname{#1}~\thmnumber{#2}\thmnote{\ -\ #3}}

\newtheoremstyle{remarkstyle}
	{\topsep} 											
	{\topsep} 											
	{}													
	{}													
	{\itshape}{}										
	{\newline}											
	{\thmname{#1}~\thmnumber{#2}\thmnote{\ -\ #3}}

\newtheoremstyle{postulatestyle}
	{\topsep} 											
	{\topsep} 											
	{}													
	{}													
	{\bfseries}{}										
	{\newline}											
	{\thmname{#1}\thmnote{ #3}}
	
\theoremstyle{normalstyle}
\newtheorem{theorem}{Theorem}[section]

\newtheorem{lemma}[theorem]{Lemma}

\newtheorem{hypothesis}[theorem]{Hypothesis}

\theoremstyle{remarkstyle}

\theoremstyle{defstyle}

\theoremstyle{postulatestyle}

\makeatletter
\DeclareRobustCommand\acrNoHyperlink[2]{%
	\begingroup
	\disableAcronymHyperlink
	#1{#2}%
	\endgroup
}

\newcommand{\disableAcronymHyperlink}{%
	\def\AC@hyperlink##1##2{##2}%
	\def\AC@hyperref[##1]##2{##2}%
	\def\AC@hypertarget##1##2{##2}%
	\def\AC@phantomsection{}%
}
\makeatother
\begin{document}
\begin{titlepage}
\pagenumbering{Alph}
\thispagestyle{empty}
\begin{center}
 
\Large\textbf{ALBERT-LUDWIGS-UNIVERSITÄT FREIBURG IM BREISGAU\\}
\vspace{0.5cm}
\Large\textbf{Institute of Mathematics}

\rule{\textwidth}{1pt}
\vspace{1.5cm}

\Large\textbf{\Title}

\vspace{1.5cm}

\includegraphics{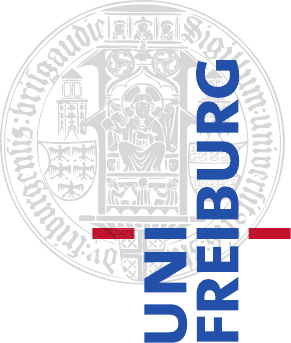}

\vfill

\normalsize
Master Thesis in Physics\\
\vspace{0.5cm}
submitted \Date\hspace{0pt} by\\
\vspace{0.5cm}
\Large\textbf{Jonas Pleyer}\\
\normalsize
\vspace{0.5cm}
\large Supervisor: JProf. Dr. Nadine Große\\
\normalsize

\newpage
\thispagestyle{empty}
My utmost gratitude goes to my parents Adi and Anke Pleyer which guided me as a child and nourished the fundamental interests in my scientific understanding of nature.\\
\vfill
I also thank Alex, Gesa, Sebastian and Yannik for the kind editorial notes.

\end{center}

\newpage\null\thispagestyle{empty}\newpage
\end{titlepage}
\thispagestyle{empty}
\tableofcontents
\pagenumbering{Roman}
\newpage
\thispagestyle{empty}
\section*{List of Abbreviations}
\begin{acronym}[TOV]
	\acro{gr}[GR]{General Relativity}
	\acro{ODE}[ODE]{Ordinary Differential Equation}
	\acro{TOV}[TOV]{Tollmann-Oppenheimer-Vollkoff}
	\acro{eos}[EoS]{Equation of State}
	\acro{LE}[LE]{Lane-Emden}
\end{acronym}

\addcontentsline{toc}{section}{\listfigurename}
\listoffigures
\addcontentsline{toc}{section}{\listtablename}
\listoftables

\section*{List of Symbols}
\begin{tabbing}
	\hspace{2cm}\=\kill
	$\Ent$				\> Entropy\\
	$\F$				\> Free Energy\\
	$\U$				\> Internal Energy\\
	$\Z$				\> Partition Function\\
	$p$					\> Pressure\\
	$\rho$				\> Density
\end{tabbing}

\section*{Units}
This thesis uses the geometrized unit system $G=c=1$.
If not otherwise mentioned any variable named $G$ or $c$ can be set to $1$ and is only written for clarification.

\newpage
\pagenumbering{arabic}
\section{Introduction}
\label{sec:01-Introduction}
The theory of \ac{gr} has allowed scientists deep insights into space and time itself.
From its emergence in 1916~\cite{einsteinGrundlageAllgemeinenRelativitaetstheorie1916a} to its current state it is regarded as one of the most elegant and well tested theories that Physics has to offer~\cite{hafeleAroundtheWorldAtomicClocks1972, vessotTestRelativisticGravitation1980a, battatApachePointObservatory2009}.
Most objects studied in this framework are typically very heavy or close to the speed of light and thus part of the large scale structure of our observable universe.
\ac{gr} is formulated in a geometric context by regarding space and time itself as a Lorentz-Manifold called spacetime with a non-degenerate, symmetric bilinear form called the metric.
The curvature of this Lorentz-Manifold, which can be completely described by the Riemann tensor $R_{\mu \nu}{ }^{\sigma}{ }_{\lambda}$, depends on mass and energy present at the respective points in spacetime which is reflected in the Einstein equations~\cite{einsteinFeldgleichungenGravitation1915}.
\begin{equation}
	G_{\mu\nu}+\Lambda g_{\mu\nu} = R_{\mu\nu} + \left(\frac{1}{2}R+\Lambda\right)g_{\mu\nu}=8\pi T_{\mu\nu}
	\label{eq:01-Intr-Einstein-Equ}
\end{equation}
The left hand side of equation~\ref{eq:01-Intr-Einstein-Equ} contains the Einstein Tensor $G_{\mu\nu}=R_{\mu\nu}+1/2Rg_{\mu\nu}$ which is defined in terms of the Ricci $R_{\mu \nu}=R_{\lambda\mu}{}^{\lambda}{}_\nu$ and metric tensor $g_{\mu\nu}$ and the Ricci scalar $R=g^{\mu\nu}R_{\mu\nu}$.
The Energy Momentum tensor on the right-hand side carries information about mass and energy.
With the Levi-Civita connection on spacetime we can write the Riemann tensor as derivatives of the metric tensor which means the left-hand side of the Einstein equations~\ref{eq:01-Intr-Einstein-Equ} contains non-linear derivatives of $g_{\mu\nu}$.
This non-linearity is one of the key distinctions to other fundamental theories such as a Quantum Theory of Fields and is one of the reasons why few exact solutions of Einsteins equations exist.
In addition, the derivatives of $g_{\mu\nu}$ are also partial which means this theory in general has to deal with non-linear partial differential equations.
A behaviour that can not be altered by simply changing variables.
While these statements might sound demotivating, lots of methods to linearize or obtain \acp{ODE} have been developed.
In particular this thesis will use suitable chosen assumptions to derive a famous \ac{ODE}, the \ac{TOV} equation and analyse its properties.\\
Of the plentiful present problems one in particular of significant importance has been to calculate the macroscopic properties such as mass and angular momentum of stellar objects which in practice are often spherically or cylindrically symmetrical.
Stars, black holes, galaxies and galaxy clusters are just a few examples of this category.
The \ac{TOV} equation~\cite{tolmanStaticSolutionsEinstein1939, oppenheimerMassiveNeutronCores1939} relates pressure, density and mass inside a spherically symmetric object
\begin{equation}
	\begin{aligned}
			\frac{\partial m}{\partial r} &= &&4\pi\rho(r)r^2\\
			\frac{\partial p}{\partial r} &= -&&\frac{Gm\rho}{r^2}\left(1+\frac{p}{\rho c^2}\right)\left(\frac{4\pi r^3 p}{mc^2}+1\right)\left(1-\frac{2Gm}{rc^2}\right)^{-1}
			\label{eq:01-Intr-TOV-Equation}
	\end{aligned}
\end{equation}
and given an \ac{eos} $\rho(p)$ one can solve the differential equation and thus obtain parameters such as the mass and total radius of said stellar object.
Together with non-negative initial values, we specifically prove, that solutions in general exist and are unique.
The non-relativisitc limit of the \ac{TOV} equation is the \ac{LE} equation~\cite{laneTheoreticalTemperatureSun1870, emdenGaskugeln1907} which can be derived in the Newtonian theory of motion when applying the condition of hydrostatic equilibrium.
The outside solution of the \ac{TOV} equation is characterised by the well known Schwarzschild solution~\cite{schwarzschildUberGravitationsfeldMassenpunktes1916} which is also the outside solution of black holes
\begin{equation}
	g=-\left(1-\frac{r_{\mathrm{s}}}{r}\right) c^{2} \diff t^{2}+\left(1-\frac{r_{\mathrm{s}}}{r}\right)^{-1} \diff r^{2}+r^{2}\left( \diff\theta^2+\sin^2\theta \diff\phi^2\right)
	\label{eq:01-Intr-Schwarzschild-Metric}
\end{equation}
and where $r_\mathrm{s}=2GM/c^2$ is the Schwarzschild radius.
Subsection~\ref{subsec:3-Mass-Mass-Bounds} will show that finding constraints for the radius of a stellar object immediately leads to a restriction on its mass.
Numerical results demonstrate that the behaviour of zero values for the \ac{TOV} equation is similar to that of the \ac{LE} equation, details of which will be discussed in section~\ref{subsec:4-NumSol-Sec-TOV-Exponents}.\\
From a Physics perspective, it is clear that Thermodynamics has to play a role in the description of stars since one has to deal with large numbers of particles well in equilibrium.
Thermodynamics is a statistical theory of nature and describes macroscopic phenomena by knowing microscopic behaviour of particles.
Carnot~\cite{carnotReflexionsPuissanceMotrice1824} outlined the first relations between a thermodynamic processes an engine and motive power.
Following in his steps many researchers such as Thomson (Lord Kelvin)~\cite{thomsonAbsoluteThermometricScale2011}, Clausius~\cite{clausiusMechanischeWaermetheorie1876}, Maxwell~\cite{maxwellScientificLettersPapers2002}, Boltzmann~\cite{boltzmannUberMechanischeBedeutung1866} and Gibbs~\cite{gibbsElementaryPrinciplesStatistical2010} greatly impacted the development of Thermodynamics.
Over the years, different ensembles have emerged which can equivalently, although in practice often simplified by a specific choice, describe the same physical system.
The canonical ensemble was first described by Ludwig Boltzmann~\cite{boltzmannUeberEigenschaftenMonocyclischer1885a}.
Its corresponding partition function is given by
\begin{equation}
	\mathcal{Z}(T,V,N) = \int\exp\left(-\frac{H(x_1,\dots,p_N)}{k_B T}\right)\frac{\diff x_1\dots \diff p_N}{N!h^{3N}}
	\label{eq:01-Intr-Canocnical-Ens-Part-Funct}
\end{equation}
where $H$ denotes the Hamiltonian of the particles.
The partition function $\mathcal{Z}$ can then be used to calculate properties such as pressure $p$, internal energy $\mathcal{U}$ and thus also the energy density $\mathcal{U}/V$.
Together with an adiabatic condition one would hope to find a relation between energy density and pressure to derive a \acp{eos}.
In many examples, a polytropic \ac{eos} between pressure $p$ and density $\rho$ is assumed
\begin{equation}
	\rho = Ap^{1+1/n}
	\label{eq:01-Intr-Poly-EOS}
\end{equation}
where $A$ is a constant.
The exponent $n=1/(\gamma-1)$ called the polytropic index can in a special case be related to the adiabatic index $\gamma_\text{ad}$ of a thermodynamic process.
In this case, equation~\ref{eq:01-Intr-Poly-EOS} is derived using the ideal gas equation $pV=Nk_B T$ and adiabaticity $\updelta Q=0$.
Another important relation is between $\gamma_\text{ad}$ and degrees of freedom.
For an ideal gas, the index can be given by $\gamma_\text{ad}=1+2/f$, where $f$ is the degrees of freedom.
In this case, $\gamma$ directly yields information about the underlying microscopic behaviour of the gas.
While there exist more recent examples of more complex \ac{eos} models~\cite{hummerEquationStateStellar1988}, this equation already covers a wide range of cases~\cite{horedtPolytropesApplicationsAstrophysics2004}.
The first part of this thesis will be concerned with deriving a fully special relativistic \ac{eos} of a non-interacting gas that will later be used in numerical solutions of the \ac{TOV} equation.
\newpage
\section{Thermodynamic Calculation of an \texorpdfstring{\acrNoHyperlink{\acs}{eos}}{EoS}}
\label{sec:2-Thermo}
This chapter aims at developing a fully special relativistic \ac{eos} of a non-interacting gas.
We briefly summarise important concepts necessary for the derivation.
In the canonical ensemble, for which an introduction can be found in~\cite{fliessbachStatistischePhysikLehrbuch2018}, the internal energy $\U$ is obtained by the relation $\U=\F-TS$ where the quantities $\F$ and $S$ can be derived by means of the partition function $\Z$ while $T$ is a variable.
Microscopically, the partition function is given by the behaviour of the $N$ particles determined by the Hamiltonian $H$.
Concepts and definitions of Hamilton Mechanics can be found in~\cite{eschrigTopologyGeometryPhysics2011, fliessbachMechanikLehrbuchZur2020, spivakPhysicsMathematiciansMechanics2010}.
In general we assume $H:\R^{3N}\times\R^{3N}\rightarrow\R$ to be a positive smooth function that assigns an energy to the positions and momenta of $N$ particles.
This explicitly takes form in the well known equations
\begin{align}
	\Z(T,V,N) 	&= \int\limits_{\mathcal{V}^N\times\R^N}\exp\left(-\frac{H(x_1,\dots,x_N,p_1,\dots,p_N)}{k_B T}\right)\frac{\diff x_1\diff p_1\dots \diff x_N\diff p_N}{N!h^{3N}}\\
	\F(T,V,N) 	&= - k_B T\log\left(\Z(T,V,N)\right)\label{2-IntEner-FreeEnerDef}\\
	\U(T,V,N) 	&= \F + T \Ent \hspace{0.2cm}=\hspace{0.2cm} \F - T\frac{\partial\F}{\partial T} \hspace{0.2cm}=\hspace{0.2cm} k_B T^2\frac{1}{\Z}\frac{\partial\Z}{\partial T}\label{2-IntEner-Def}
\end{align}
where $x_i\in \mathcal{V}\subset\R^3$  and $p_i\in \R^3$.
The relation $V=\text{vol}(\mathcal{V})$ describes the total volume occupied by the medium while $T$ and $N$ are its temperature and particle number respectively.
To fully calculate $U$ it is necessary to obtain the partition function $\Z$.
From there $\U/V$ as the energy density can be compared to the pressure defined by
\begin{equation}
	p=-\frac{\partial\F}{\partial V}
	\label{eq:2-IntEner-pressureDef}
\end{equation}
to yield an \ac{eos}.
\begin{subsection}{The Ultra-Relativistic Internal Energy}
The ultra-relativistic case is a textbook example and well known limit that we will use in order to verify our later calculated results. 
First we write down the ultra-relativistic Hamiltonian given by 
\begin{equation}
	H(x,p)=||p||c.
	\label{eq:2-IntEner-Ultra-Rel-Hamiltonian}
\end{equation}
The corresponding partition function for an $N$ particle system then reads ($\beta=(k_B T)^{-1}$)
\begin{align}
	\Z 	&= \frac{V^N}{N!h^{3N}}\left[\int\limits_{\R^3}\exp\left(- \beta H(p)\right)\diff^3 p\right]^N\\
		&= \frac{V^N}{N!h^{3N}}\left[\int\limits_0^\infty4\pi p^2\exp(-\beta pc) \diff p\right]^N\\
		&= \frac{V^N}{N!h^{3N}}\frac{(4\pi)^N}{(\beta c)^{3N}}\left[\int\limits_0^\infty x^2\exp(-x) \diff x\right]^N\label{2-IntEner-UltraRelZ-Int-eq}\\
	\Z	&= \frac{1}{N!}\left(8\pi V\left(\frac{k_B T}{hc}\right)^3\right)^N
	\label{eq:2-IntEner-UltraRelZ}
\end{align}
where from the first to second line we used usual spherical coordinates and afterwards the integral transformation $x=\beta cp$.
The integral in equation~\eqref{2-IntEner-UltraRelZ-Int-eq} can be solved exactly with value $2$.
In the last line $\beta=(k_B T)^{-1}$ was used for visual clarity.
With equation~\eqref{2-IntEner-Def} the internal energy $\U = 3Nk_B T$ and with equation~\eqref{eq:2-IntEner-pressureDef} the \ac{eos} can now be written down
\begin{equation}
	p = \frac{Nk_B T}{V} = \frac{1}{3}\frac{\U}{V} = \frac{1}{3}\rho.
	\label{eq:2-IntEner-eos-nonrel}
\end{equation}
\end{subsection}
\begin{subsection}{The Special-Relativistic Internal Energy}
A source for the following calculations could not be found and were thus carried out by the author solely.
The special relativistic Hamiltonian is given by
\begin{equation}
	H=mc^2\sqrt{1+\frac{p^2}{m^2 c^2}}.
	\label{eq:2-IntEner-Rel-Hamiltonian}
\end{equation}
The ultra-relativistic limit can be obtained by letting $m\rightarrow0$.
In this limiting case we should be able to recover the results from equation~\eqref{eq:2-IntEner-UltraRelZ}.
\begin{align}
	\Z 	&= \frac{V^N}{N!h^{3N}}\left[\int\limits_{\R^3}\exp\left(- \frac{mc^2\sqrt{1+\frac{p^2}{m^2 c^2}}}{k_B T}\right)\diff^3 p\right]^N\label{2-IntEner-EQ1}\\
		&= \frac{V^N}{N!h^{3N}}\left[\int\limits_0^\infty 4\pi p^2\exp\left(-\beta mc^2\sqrt{1+\frac{p^2}{m^2 c^2}} \right)\diff p\right]^N\\
		&= \frac{(4\pi V)^N}{N!}\left(\frac{mc}{h}\right)^{3N}\left[\int\limits_0^\infty q^2\exp\left(-\alpha\sqrt{1+q^2}\right)\diff q\right]^N\\
		&= \frac{(4\pi V)^N}{N!}\left(\frac{mc}{h}\right)^{3N}\left(\int\limits_0^\infty\sinh(x)^2\cosh(x)\exp(-\alpha\cosh(x))\diff x \right)^N\\
		&= \frac{(4\pi V)^N}{N!}\left(\frac{mc}{h}\right)^{3N}\left(\int\limits_0^\infty\frac{\sinh(2x)}{2}\sinh(x)\exp(-\alpha\cosh(x))\diff x \right)^N\\
		&= \frac{(4\pi V)^N}{N!}\left(\frac{mc}{h}\right)^{3N}\Biggl(\left.-\frac{\sinh(2x)}{2\alpha}\exp(-\alpha\cosh(x))\right|_0^\infty\\
		&\hspace{12em} + \frac{1}{\alpha}\int\limits_0^\infty\cosh(2x)\exp(-\alpha\cosh(x))\diff x  \Biggr)^N\\
		&= \frac{(4\pi V)^N}{N!}\left(\frac{mc}{h}\right)^{3N}\left(\frac{1}{\alpha}\int\limits_0^\infty\cosh(2x)\exp(-\alpha\cosh(x))\diff x  \right)^N\\
		&= \frac{1}{N!}\left(8\pi V\left(\frac{k_B T}{hc}\right)^{3}\frac{\alpha^2 K_2(\alpha)}{2}\right)^N
		\label{2-IntEner-PartFunc}
\end{align}
In the first step we used spherical coordinates followed by the substitution $qmc=p$ and $\alpha=\beta mc^2=mc^2/k_B T$.
Afterwards we substituted $q=\sinh(x)$ and used the identity $\cosh(x)\sinh(x)=\sinh(2x)/2$.
Partial integration then leads to the last integral which can be identified as the modified Bessel function of the 2nd kind $K_2(\alpha)$~\cite{abramowitzPocketbookMathematicalFunctions1984}.
The equation is then rewritten such that the ultra-relativistic limit can be read off upon letting $\alpha\rightarrow0$.\\
We can now calculate the internal energy $\U$ from $\Z$ via equation~\eqref{2-IntEner-Def}
\begin{align}
    \U &= 3Nk_B T - Nk_B T\alpha\left(\frac{\partial_\alpha K_2(\alpha)}{K_2(\alpha)}+2\right)\\
    \U &= 3Nk_B T - mc^2\left(\frac{\partial_\alpha K_2(\alpha)}{K_2(\alpha)}+2\right).
    \label{2-IntEner-InternelEnergyExplicit}
\end{align}
Again, it can be seen that the ultra-relativistic limit can be obtained by letting $\alpha\rightarrow0$, since the term written in parenthesis vanishes.
It becomes immediately clear that the ultra-relativistic \ac{eos} is $\rho=3p$ which is equivalent to a traceless energy momentum tensor.\\
From equation~\eqref{2-IntEner-PartFunc} and~\eqref{2-IntEner-FreeEnerDef}, we immediately derive the ideal gas equation via the definition of pressure (see equation~\eqref{eq:2-IntEner-pressureDef} in the canonical ensemble
\begin{equation}
    p=\frac{\partial\F}{\partial V} = \frac{Nk_ BT}{V}.
    \label{eq:2-IntEner-IdealGasEq}
\end{equation}
Note that in the case of the Hamiltonian $H$ having a non-trivial dependence on the coordinates $x_i$, the partition function would differ by a factor involving $V$ and $T$.
This more general scenario does not yield the ideal gas law and would thus need to be treated separately.
\end{subsection}
\begin{subsection}{The Special-Relativistic Equation of State}
\label{2-IntEner-SR-EOS-Derivation}
This section aims to develop an equation between the thermodynamic energy density $\rho=\U/V$ and the pressure $p$ of the gas given by the ideal gas equation~\eqref{eq:2-IntEner-pressureDef}.
We assume an additional constraint namely adiabaticity and thus further reduce the degrees of freedom of the thermodynamic system.\\
When assuming an adiabatic condition $\updelta Q=0$ and using the First Law of Thermodynamics~\cite{fliessbachStatistischePhysikLehrbuch2018} $\diff U =\updelta Q + \updelta W$, where $\updelta W=-p\diff V$ and $\diff U=C_V\diff T$, we can relate pressure and temperature. 
This adiabatic condition is quite well satisfied since the timescale of radiational and other losses compared to the timescale of thermodynamic events is negligible~\cite{noerdlingerSolarMassLoss2008, vinkMassLossStellar2017}.
Using equation~\eqref{2-IntEner-InternelEnergyExplicit} and~\eqref{eq:2-IntEner-pressureDef}, we obtain
\begin{alignat}{3}
    -p\diff V &= C_V\diff T\\
    -\frac{Nk_B T}{V}\diff V &= Nk_B &\Biggl[1 + &\alpha^2\left(\left(\frac{\partial_\alpha K_2(\alpha)}{K_2(\alpha)}\right)^2 - \frac{\partial^2_\alpha K_2(\alpha)}{K_2(\alpha)}\right)\Biggr]\diff T\\
    - \frac{\diff V}{V} &= &\Biggl[1 - &\alpha^2\partial_\alpha\left(\frac{\partial_\alpha K_2(\alpha)}{K_2(\alpha)}\right)\Biggr]\frac{\diff T}{T}\\
    &= &\Bigl[1 - &\alpha^2\partial_\alpha^2\left(\log K_2(\alpha)\right)\Bigr]\frac{\diff T}{T}
\end{alignat}
This equation also shows explicitly the $T$ dependence of the specific heat $C_V$ for the general case.
Again, taking the ultra-relativistic limit by letting $\alpha\rightarrow0$, the right-hand term in the first equation converges to $-2Nk_B$.
This agrees with the expected specific heat for an ultra-relativistic gas $C_{V,\text{ur}}=3Nk_B$.
With the identity $\diff\alpha/\alpha = -\diff T/T$ (using $\alpha=mc^2/k_B T$), we can transform the equation and integrate it.
After applying partial integration, the result is
\begin{align}
	\frac{\diff V}{V} &= \left(1-\alpha^2\partial_\alpha^2\log K_2(\alpha)\right)\frac{\diff\alpha}{\alpha}\\
    \log\left(\frac{V}{V_0}\right) &= \log\left(\frac{\alpha}{\alpha_0}\right) - \int\limits_{\alpha_0}^\alpha \alpha\frac{\partial^2}{\partial \alpha^2}\log(K_2(\alpha'))\diff\alpha'\\
    &=\log\left(\frac{\alpha}{\alpha_0}\right) + \log\left(\frac{K_2(\alpha)}{K_2(\alpha_0)}\right) - \Bigl[\alpha\frac{\partial_\alpha K_2}{K_2}\Bigr]_{\alpha_0}^\alpha
\end{align}
This equation enables us to write down a relation between volume and temperature (encapsulated in $\alpha=mc^2/k_B T$)
\begin{equation}
	V(\alpha) = \frac{\alpha K_2(\alpha)}{C}\exp\left(\alpha\frac{K_3(\alpha)+K_1(\alpha)}{2K_2(\alpha)}\right)
	\label{eq:2-IntEner-Volume-Alpha-Dependence}
\end{equation}
where the constant $C$ is defined by the equation beforehand and only depends on the integration boundaries $\alpha_0$ and $V_0$.
It is given by
\begin{equation}
	C = \frac{\alpha_0 K_2(\alpha_0)}{V_0}\exp\left(\alpha_0\frac{K_1(\alpha_0)+K_3(\alpha_0)}{2K_2(\alpha_0)}\right).
	\label{eq:2-IntEner-Parameter-Eos-Integration-Result}
\end{equation}
Since the goal of this section is to obtain a readable output for an \ac{eos}, it is necessary to construct a bijection relating $p$ and $T$.
This becomes clear when writing down the energy density
\begin{equation}
	\rho = \frac{\U}{V} = \frac{Nk_B T}{V} - \frac{Nk_B T}{V}\left(\alpha\frac{\partial_\alpha K_2(\alpha)}{K_2(\alpha)}\right)
	\label{eq:2-IntEner-DensityAlpha}
\end{equation}
where $p=Nk_B T/V$ can be easily identified but the $T$ dependence via $\alpha$ is not solved yet.\\
The pressure $p$ can be rewritten to take the form
\begin{equation}
	p = \frac{Nk_B T}{V} = CNmc^2\frac{1}{K_2(\alpha)\alpha^2}\exp\left(-\alpha\frac{K_1(\alpha)+K_3(\alpha)}{2K_2(\alpha)}\right).
	\label{eq:2-IntEner-PressureAlpha}
\end{equation}
At this point it is not reasonable to ask what happens in the ultra-relativistic limit since $C$ depends non-trivially on $m$ and thus $m$ is not fully substituted in $\alpha$.\\
Interestingly, the pressure seems to be constant for very high temperatures.
The limiting case is obtained when taking $T\rightarrow\infty$ (which corresponds to $\alpha\rightarrow0$)
\begin{equation}
	\lim\limits_{\alpha\to0}\left[\frac{1}{K_2(\alpha)\alpha^2}\exp\left(-\alpha\frac{K_1(\alpha)+K_3(\alpha)}{2K_2(\alpha)}\right)\right] = \frac{1}{2\e^2} \approx 0.006767.
	\label{eq:2-IntEner-PressureAlpha-Limit}
\end{equation}
The same argument then holds true for the density given by equation~\eqref{eq:2-IntEner-DensityAlpha} and since
\begin{equation}
	\lim\limits_{\alpha\to0}\left[1+\alpha\frac{K_1(\alpha)+K_3(\alpha)}{2K_2(\alpha)}\right] = 3
	\label{eq:2-IntEner-Limit-Alpha}
\end{equation}
we have
\begin{equation}
	\lim\limits_{\alpha\to0}\left(\frac{\rho(\alpha)}{Cmc^2}\right) = \frac{3}{2\e^2} \approx 0.203003.
	\label{eq:2-IntEner-Limit-Alpha-2}
\end{equation}
\begin{theorem}
	The mapping $p:\R_{>0}\rightarrow \R_{>0},\alpha\mapsto p(\alpha)$ written down in equation~\eqref{eq:2-IntEner-PressureAlpha} is a bijection for any $N,m,c,C\neq0$.
	\begin{equation}
		p = \frac{Nk_B T}{V} = CNmc^2\frac{1}{K_2(\alpha)\alpha^2}\exp\left(-\alpha\frac{K_1(\alpha)+K_3(\alpha)}{2K_2(\alpha)}\right)
		\label{eq:2-IntEner-Pressure-Alpha-Dependence}
	\end{equation}
\end{theorem}
\begin{proof}
	For this proof it suffices to show that the function $p(\alpha)$ has a strictly monotonous behaviour.
	Without loss of generality, we assume $N,m,c,C>0$.
	Now it is obvious that the first two terms $Nmc^2/CK_2(\alpha)$ and $\alpha^{-2}$ are strictly decreasing.
	This is easy to see when using~\cite{abramowitzPocketbookMathematicalFunctions1984}
	\begin{equation}
		\frac{\partial K_n}{\partial\alpha} = \frac{n}{\alpha}K_n-K_{n+1} = -\frac{K_{n-1}+K_{n+1}}{2}.
		\label{eq:2-IntEner-BesselDerivative}
	\end{equation}
	We then calculate the derivative of the third term and divide by the exponential (since it is positive)
	\begin{align}
		&\hspace{1em} \frac{1}{\exp(\alpha\partial_\alpha\log(K_2))}\frac{\partial}{\partial\alpha}\exp(\alpha\partial_\alpha\log(K_2))\\
		&= \partial_\alpha\log(K_2) + \alpha\partial^2_\alpha\log(K_2)\\
		&= \frac{\partial_\alpha K_2}{K_2} + \alpha\frac{\partial^2_\alpha K_2}{K_2} + \alpha\left(\frac{\partial_\alpha K_2}{K_2}\right)^2\\
		&= \frac{K_1-\frac{2}{\alpha}K_2}{K_2} + \alpha\frac{\frac{1}{\alpha}K_1-K_2-\frac{2}{\alpha}\left(K_1-\frac{2}{\alpha}K_2\right)+\frac{2}{\alpha^2}K_2}{K_2^2}\\
		&\hspace{5cm}+\alpha\frac{K_1^2-\frac{1}{\alpha}K_1 K_2-\frac{4}{\alpha^2}K_2^2}{K_2^2}\\
		&= -\alpha + \alpha\frac{K_1^2}{K_2^2} - 4\frac{K_1}{K_2}
	\end{align}
	thus, it is sufficient to show that
	\begin{equation}
		\alpha\frac{K_1^2}{K_2^2}  < \alpha + 4\frac{K_1}{K_2}.
		\label{eq:2-IntEner-BesselK-Ineq-K1-over-K2}
	\end{equation}
	We quickly prove the more general result $K_\nu<K_{\nu+1}$.
	One possible definition~\cite{abramowitzPocketbookMathematicalFunctions1984} for the Bessel function $K_\nu$ is given by
	\begin{equation}
		K_\nu \coloneqq \frac{\sqrt{\pi}}{\left(\nu-\frac{1}{2}\right)!}\left(\frac{1}{2}z\right)^\nu\int\limits_1^\infty\e^{-tz}\left(t^2-1\right)^{\nu-\frac{1}{2}}\diff t
		\label{eq:2-IntEner-BesselK-Def}
	\end{equation}
	We inspect the ratio
	\begin{equation}
		\frac{K_\nu}{K_{\nu+1}} = \frac{1}{\left(\nu+\frac{1}{2}\right)\left(\frac{1}{2}z\right)}\frac{\int\e^{-tz}\left(t^2-1\right)^{\nu-1/2}\diff t}{\int\e^{-tz}\left(t^2-1\right)^{\nu+1/2}\diff t}
		\label{eq:2-IntEner-BesselK-Ratio}
	\end{equation}
	and rewrite the denominator with partial integration
	\begin{equation}
		 \frac{1}{2}z\int\limits_1^\infty\e^{-tz}\left(t^2-1\right)^{\nu+1/2}\diff t = \left(\nu+\frac{1}{2}\right)\int\limits_1^\infty \e^{-tz}t\left(t^2-1\right)^{\nu-1/2}\diff t.
		 \label{eq:2-IntEner-BesselK-Integral-Evaluation}
	\end{equation}
	Now it is obvious that $K_{\nu+1}>K_\nu$.
	Thus in total, the function given by equation~\eqref{eq:2-IntEner-FinalEOS} can be inverted.
\end{proof}%
\noindent With this mapping $p:\R_{>0}\rightarrow \R_{>0},\alpha\mapsto p(\alpha)$ and its inverse $\alpha:\R_{>0}\rightarrow \R_{>0},p\mapsto \alpha(p)$, we can use~\eqref{eq:2-IntEner-DensityAlpha} and finally write down the \ac{eos}
\begin{equation}
	\rho = \frac{\U}{V} = p\left(1+\alpha(p)\frac{K_1 (\alpha(p))+K_3(\alpha(p))}{2K_2 (\alpha(p))}\right).
	\label{eq:2-IntEner-FinalEOS}
\end{equation}
Figure~\ref{fig:2-IntEner-RelEOSPlot} is obtained by choosing numerical values and then interpolating and inverting equation~\eqref{eq:2-IntEner-PressureAlpha}.
The constant factor
\begin{equation}
	B\coloneqq CNmc^2/p_0
	\label{eq:2-IntEner-FactorExplanation}
\end{equation}
is substituted to obtain independence of $p_0$\footnote{This will be the starting value for the pressure in the \acl{TOV} equation. Since the pressure decreases from the inside a star to the outside, this resembles our highest value of $p$.}.
Furthermore the graphs of the plotted \ac{eos} are normalised such that
\begin{equation}
	\rho_{i}(p_0)=\rho_{0,i}
	\label{eq:2-IntEner-Plt-Initial-Val-Def}
\end{equation}
and thus can be compared with each other.
Note that the $\rho_0$ of the plot is not a universal value but rather each function has been scaled individually and $\rho_0$ is a placeholder for the corresponding $\rho_{0,i}$.
The graph of the polytropic \ac{eos} $\rho_{cla}(p)=Ap^{1/\gamma}$ in figure~\ref{fig:2-IntEner-RelEOSPlot} when normalised as before, is independent of $A$ and can be uniquely characterised by $n$.
For further details, the interested reader is referred to~\cite{pleyerGithubRepositoryJonas2021}.
\begin{figure}[H]
	{\centering
	\import{pictures/2-InternalEnergy-Explicit/}{RelEOS.pgf}
	}
	\caption{Relativistic \acl{eos}}
	\label{fig:2-IntEner-RelEOSPlot}
	\small
	The relativistic \ac{eos} $\rho(p)$ is normalised such that values can be compared with a polytropic \ac{eos}.
	Graphs for the relativistic version are independent of the exponent $n$ which is a degree of freedom intrinsic to the polytropic \ac{eos}.
	By normalisation, the graphs of the polytropic \ac{eos} are independent of the factor $A$.
\end{figure}
\end{subsection}
\newpage
\section{Calculating the Mass of a Star with an \texorpdfstring{\acrNoHyperlink{\acs}{eos}}{EoS}}
\label{sec:3-Mass}
\subsection{Deriving the \texorpdfstring{\acrNoHyperlink{\acs}{TOV}}{TOV}-Equation}
\label{subsec:3-Mass-Sec-TOVDerivation}
This chapter derives the \ac{TOV} equation which was introduced earlier in section~\ref{sec:01-Introduction}.
Information about \ac{gr} can be found in~\cite{choquet-bruhatAnalysisManifoldsPhysics2000, choquet-bruhatGeneralRelativityEinstein2009, choquet-bruhatIntroductionGeneralRelativity2015, waldGeneralRelativity1984}.
We follow the derivation of Wald.
We consider a spherical-symmetric static Lorentz-Manifold $(V,g)$ with charts such that the metric $g$ can be written as
\begin{equation}
	g=-\e^\nu \diff t^2+\e^\lambda \diff r^2 + r^2(\diff\theta^2+\sin^2\theta \diff\phi^2).
\end{equation}
The stress-energy tensor of an ideal fluid with density $\rho$ and pressure $p$ is given by
\begin{equation}
	T_{\mu\nu}=\rho u_\mu u_\nu + p(g_{\mu\nu}+u_\mu u_\nu)
\end{equation}
where $u$ is the 4-velocity of the fluid.
In the rest frame where $u^\mu=(-\e^{-\nu/2},0,0,0)$, this equation simplifies to
\begin{equation}
	T^\mu_\nu=\text{diag}(-\rho,p,p,p)
\end{equation}
The Christoffel symbols for this metric are
\begin{align}
	\Gamma_{\mu\nu}^0 &= \begin{bmatrix}
	                     	0 & \nu'/2 & & \\
	                     	\nu'/2 & 0 & & \\
	                     	& & 0 &\\
	                     	& & & 0
	                     \end{bmatrix}, \hspace{0.3cm}
	\Gamma_{\mu\nu}^1 = \begin{bmatrix}
	                     	\nu'\e^{\nu-\lambda}/2 & & &\\
	                     	& \lambda'/2 & & \\
	                     	& & -r\e^{-\lambda} & \\
	                     	& & & -r\sin^2\theta\e^{-\lambda}
	                     \end{bmatrix}\\
	\Gamma_{\mu\nu}^2 &= \begin{bmatrix}
	                     	0 & & &\\
	                     	& 0 & 1/r &\\
	                     	& 1/r & 0 &\\
	                     	& & & -\sin\theta\cos\theta
	                     \end{bmatrix}, \hspace{0.3cm}
	\Gamma_{\mu\nu}^3 = \begin{bmatrix}
	                     	& 0 & 0 &\\
	                     	& 0 & 0 & 1/r\\
	                     	& 0 & 0 & \cos\theta/\sin\theta\\
	                     	& 1/r & \cos\theta/\sin\theta & 0
						\end{bmatrix}
\end{align}
From these, the non-zero components of the Ricci-Tensor can be calculated 
\begin{align}
	R_{11} &= \frac{1}{4r}\e^{-\lambda}\left[\left(2r\nu''+r\nu'^2\right) + \left(4-r\lambda'\right)\nu'\right]\\
	R_{22} &= -\frac{1}{4r}\e^{-\lambda}\left[\left(2r\nu''\right)+r\nu'^2-r\lambda'\nu'-4\lambda'\right]\\
	R_{33} &= -\frac{1}{2r^2}\e^{-\lambda}\left(r\nu'-r\lambda'-2\e^\lambda+2\right)\\
	R_{44} &= R_{33}
\end{align}
and with $R_{\mu\nu}-g_{\mu\nu}R/2=G_{\mu\nu}=8\pi T_{\mu\nu}$ we ultimately receive the following field equations.
\begin{align}
	-8\pi T_0^0 = 8\pi\rho &= \frac{\lambda'\e^{-\lambda}}{r}+\frac{1-\e^{-\lambda}}{r^2}\label{eq:3-Mass-Equ-EinstEqu-1}\\
	8\pi T_1^1 = 8\pi p &= \nu'\frac{\e^{-\lambda}}{r} - \frac{1-\e^{-\lambda}}{r^2}\label{eq:3-Mass-Equ-EinstEqu-2}\\
	8\pi T_2^2 = 8\pi p &= \frac{\e^{-\lambda}}{2}\left[\nu''+\left(\frac{\nu'}{2} + \frac{1}{r}\right)\left(\nu'-\lambda'\right) \right]\label{eq:3-Mass-Equ-EinstEqu-3}
\end{align}
Since $R_4^4=R_3^3$, we omitted the last equation.
From equation \eqref{eq:3-Mass-Equ-EinstEqu-1} we infer the relation.
\begin{equation}
	\e^{-\lambda} = 1 - \frac{2}{r}\int\limits_0^r 4\pi\rho(r') r'^2\diff r' = 1 - \frac{2m(r)}{r}.
	\label{eq:3-Mass-Equ-MassDepOnr}
\end{equation}
The metric needs to be defined at every point in space and thus can not have any additional integration constant in equation~\eqref{eq:3-Mass-Equ-MassDepOnr}, since otherwise we would obtain a term $a/r$ which is not defined for $r\rightarrow0$.\\
The property $m(r)$ can be recognised as the Newtonian mass of the star (which is different to the proper mass~\cite{waldGeneralRelativity1984}).
Since $\e^{-\lambda}>0$, we immediately see that $m(r)<r/2$.\\
In addition to the Field equations \eqref{eq:3-Mass-Equ-EinstEqu-1} through~\eqref{eq:3-Mass-Equ-EinstEqu-3} the divergence of the Stress-Energy Tensor yields more information
\begin{equation}
	\nabla_\mu T^{\mu\nu}=0.
\end{equation}
The following explicit calculation, which again assumes spherical symmetry, shows how to obtain this additional restriction on the pressure and density.
\begin{align}
	\nabla_\mu T^\mu_\nu 	&= \partial_\mu T^\mu_1 + \Gamma^\mu_{\mu\sigma}T^\sigma_\nu-\Gamma^\sigma_{\mu\nu}T^\mu_\sigma\\
	\nabla_\mu T^\mu_1		&= \frac{\partial p}{\partial r} + p\left(\Gamma^0_{01}+\Gamma^1_{11}+\Gamma^2_{21}+\Gamma^3_{31} \right) - \Gamma^\sigma_{\mu 1}T^\mu_\sigma\\
							&= \frac{\partial p}{\partial r} + p\left(\frac{\nu'+\lambda'}{2} + \frac{2}{r}\right) + \rho\frac{\nu'}{2} - p\frac{\lambda'}{2} - p\frac{2}{r}\\
	\frac{\partial p}{\partial r} &= -\frac{p+\rho}{2}\nu'
\end{align}
Together with equation~\eqref{eq:3-Mass-Equ-EinstEqu-2} and the definition~\eqref{eq:3-Mass-Equ-MassDepOnr}, we can write
\begin{align}
	\frac{\partial p}{\partial r} 	&= -\frac{p+\rho}{2}\left(\frac{8\pi pr + \frac{1-\e^{-\lambda}}{r}}{\e^{-\lambda}} \right)\\
									&= -\frac{p+\rho}{2r}\left(\frac{8\pi pr+ \frac{2m}{r^2}}{1-\frac{2m}{r}} \right)\\
									&= -\frac{m\rho}{r^2}\left(1+\frac{p}{\rho}\right)\left(\frac{4\pi r^3 p}{m}+1\right)\left(1-\frac{2m}{r}\right)^{-1}\\
	\frac{\partial p}{\partial r} 	&= -\frac{Gm\rho}{r^2}\left(1+\frac{p}{\rho c^2}\right)\left(\frac{4\pi r^3 p}{mc^2}+1\right)\left(1-\frac{2Gm}{rc^2}\right)^{-1}
	\label{eq:3-Mass-Equ-TOV-Eq-1}
\end{align}
where in the last step the constants $c=G=1$ were put back in.
Equation~\eqref{eq:3-Mass-Equ-TOV-Eq-1} together with~\eqref{eq:3-Mass-Equ-MassDepOnr} yields the \ac{TOV} differential equations
\begin{alignat}{3}
	\frac{\partial m}{\partial r} &= &&4\pi\rho(r)r^2\\
	\frac{\partial p}{\partial r} &= -&&\frac{Gm\rho}{r^2}\left(1+\frac{p}{\rho c^2}\right)\left(\frac{4\pi r^3 p}{mc^2}+1\right)\left(1-\frac{2Gm}{rc^2}\right)^{-1}
	\label{3-Mass-TOV-Eq}
\end{alignat}
\subsection{Newtonian Limit}
\label{subsec:3-Mass-Sec-LEDerivation}
This section follows~\cite{weissteinLaneEmdenDifferentialEquation2020} and~\cite[89\psqq]{chandrasekharChandrasekharAnIntroductionStudy1958}.
Together with a polytropic \ac{eos} $p=K\rho^{1+1/n}$ and the definition $\rho=\lambda\theta^n$, we expect to obtain the Newtonian behavior in the non-relativistic limit in the form of the \ac{LE} equation
\begin{equation}
	\frac{K(n+1)\lambda^{1/n-1}}{4\pi}\Delta\theta+\theta^n=0.
\end{equation}
The usual non-relativistic limit is obtained from a Taylor expansion of equation \eqref{3-Mass-TOV-Eq} around $1/c^2$ in lowest order.
The resulting equation then reads
\begin{equation}
	\frac{\partial p}{\partial r} = -\frac{Gm\rho}{r^2} + \mathcal{O}\left(\frac{1}{c^2}\right).
	\label{eq:3-Mass-Equ-TOVNonRel-Limit}
\end{equation}
One could be tempted to make the assumption that $\frac{\partial p_{\text{TOV}}}{dr} \leq \frac{\partial p_{\text{LE}}}{dr}.$ However this equation fails since the mass given in equation \eqref{3-Mass-TOV-Eq} is not the same as the one given in the \ac{LE} equation. Using the previous relations for $\rho$ and $p$ and again setting $G=c=1$ , we can calculate
\begin{equation}
	\frac{\partial p}{\partial r} = \frac{\partial}{\partial r}\left(K\rho^{1+1/n}\right)= K\lambda^{1+1/n}(n+1)\theta^n\frac{\partial\theta}{\partial r} = -\frac{m\lambda\theta^n}{r^2}
\end{equation}
by using the definition of our polytropic \ac{eos} and equation \eqref{eq:3-Mass-Equ-TOVNonRel-Limit}.
Rearranging and taking the derivative of this equation and using $\partial m/\partial r = 4\pi\rho r^2$, we obtain
\begin{equation}
	- \frac{\partial m}{\partial r} = K\lambda^{1/n}(n+1)\frac{\partial}{\partial r}\left(r^2\frac{\partial\theta}{\partial r}\right) = -4\pi r^2\lambda\theta^n
\end{equation}
Upon redefining $\xi=r/\kappa$ where $4\pi\kappa^2=(n+1)K\lambda^{1/n-1}$, one can obtain the mathematically cleaner looking equation
\begin{equation}
	\frac{1}{\xi^2}\frac{\partial}{\partial\xi}\left(\xi^2\frac{\partial\theta}{\partial\xi}\right) + \theta^n=0
	\label{eq:3-Mass-Equ-Lane-Emden-Eq}
\end{equation}
Exact solutions are known for the cases $n=0,1,5$.
The derivation can be found in appendix~\ref{99-App-A-Exact-LE-Solutions}.
Figure~\ref{fig:3-Mass-Plt-LE-Exact-Results-Plots} and Table~\ref{fig:3-Mass-Tbl-LE-Exact-Results} summarise them.
Together with equation~\eqref{eq:3-Mass-Equ-Lane-Emden-Eq}, we can already suspect that for values $n\geq5$, the equation does not yield solutions with a zero value.\\
\noindent
\begin{minipage}{0.5\textwidth}
	\centering
	\import{pictures/3-MassOfStar/}{LE-SingleSolve.pgf}
\end{minipage}\hfill%
\begin{minipage}{0.45\textwidth}
	\renewcommand{\arraystretch}{1.2}
	\begin{tabular}[b]{@{}lcccc@{}}
		\toprule
		$n$ & \phantom{ab} & \ac{LE} Solution & \phantom{ab} & $\xi_0$\\
		\cmidrule{1-5}
		$n=0$ && $\displaystyle 1-\frac{1}{6}\xi^2$ && $\sqrt{6}$\\[3ex]
		$n=1$ && $\displaystyle \frac{\sin(\xi)}{\xi}$ && $\pi$\\[3ex]
		$n=5$ && $\displaystyle \frac{1}{\sqrt{1+\frac{1}{3}\xi^2}}$ && $\infty$\\[2ex]
		\bottomrule
	\end{tabular}
\end{minipage}
\begin{minipage}[t]{0.5\textwidth}
	\begin{figure}[H]
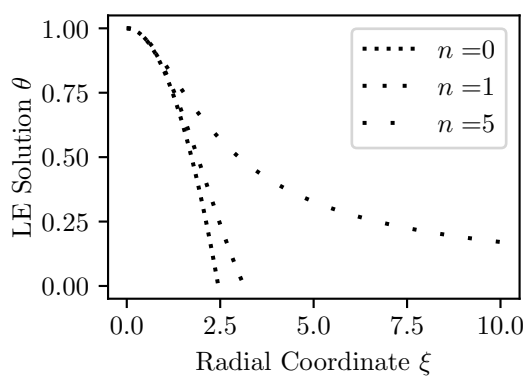

		\caption[Graph of exact Lane Emden Solutions]{Graph of exact \ac{LE} solutions.}
		\label{fig:3-Mass-Plt-LE-Exact-Results-Plots}
	\end{figure}
\end{minipage}\hfill%
\begin{minipage}[t]{0.45\textwidth}
	\begin{figure}[H]
		\captionof{table}[Lane Emden exact Solutions]{\ac{LE} exact Solutions and their zero value.}
		\label{fig:3-Mass-Tbl-LE-Exact-Results}
	\end{figure}
\end{minipage}
\subsection{Mass Bounds}
\label{subsec:3-Mass-Mass-Bounds}
We will first follow the approach given in \cite{waldGeneralRelativity1984}. 
The first assumptions will be $d\rho/dr<0$ and $\rho\geq0$. 
Also we again consider a compact star, meaning $\rho(r)=0$ for all $r>R$.
We first state a useful Lemma for the proof of our next theorem.
\begin{lemma}
	Given $\rho\in C^1(\R_{\geq0})$ is monotonously decreasing, the function $\rho_{av}=m/r^3$ has negative slope.
\end{lemma}
\begin{proof}
	Taking the derivative of $\rho_{av}$ and applying equation~\eqref{eq:3-Mass-Equ-MassDepOnr}, we obtain
	\begin{equation}
		\frac{\partial}{\partial r}\left(\frac{m}{r^3}\right) = -3\frac{m}{r^4} + \frac{4\pi\rho}{r}.
	\end{equation}
	By equation~\eqref{3-Mass-TOV-Eq} we know that $p$ has negative slope and since $\rho$ is monotonously increasing, we have
	\begin{equation}
		m = \int\limits_0^r 4\pi\rho r^2 \diff r \geq \rho\int\limits_0^r4\pi r^2 \diff r= \frac{4\pi r^3\rho}{3}
	\end{equation}
	which shows that $\frac{4\pi\rho r^3}{3}\leq m$ and completes the proof.
\end{proof}\noindent
\begin{theorem}[Mass Bound]
	The Mass of a spherically symmetric star is bound from above by
	\begin{equation}
		M < \frac{4}{9}R.
	\end{equation}
\end{theorem}
\begin{proof}
	In our attempt to obtain an upper limit for the mass of a spherically symmetric star, we start by taking the difference of equation \eqref{eq:3-Mass-Equ-EinstEqu-2} and \eqref{eq:3-Mass-Equ-EinstEqu-3}, we obtain
	\begin{align}
		0 &= \nu'\frac{\e^{-\lambda}}{r} - \frac{1-\e^{-\lambda}}{r^2} - \frac{\e^{-\lambda}}{2}\left[\nu''+\left(\frac{\nu'}{2} + \frac{1}{r}\right)\left(\nu'-\lambda'\right) \right]\\
		&= - \frac{2m(r)}{r^3} + \frac{\lambda'\e^{-\lambda}}{2r} -  \frac{\e^{-\lambda}}{2}\left[\nu''+ \frac{\nu'^2}{2} - \frac{\nu'}{r} -\frac{\lambda'\nu'}{2}\right]\\
		&= r\frac{\partial}{\partial r}\left(\frac{m(r)}{r^3}\right) - \frac{\e^{-\lambda}}{2}\left[\nu''+ \frac{\nu'^2}{2} - \frac{\nu'}{r} -\frac{\lambda'\nu'}{2}\right]\\
		0 &= \frac{\partial}{\partial r}\left(\frac{m(r)}{r^3}\right) - \frac{\e^{-\lambda}}{2}\left[\frac{\nu''}{r}+ \frac{\nu'^2}{2r} - \frac{\nu'}{r^2} -\frac{\lambda'\nu'}{2r}\right]\\
		&= \frac{\partial}{\partial r}\left(\frac{m(r)}{r^3}\right) - \frac{1}{2}\e^{-\frac{\lambda+\nu}{2}}\frac{\partial}{\partial r}\left[\frac{1}{r}\nu'\e^{\frac{\nu-\lambda}{2}}\right].
	\end{align}
	Since $\partial_r\rho\leq0$, also the average density $m(r)/r^3$ decreases with $r$.
	Thus we obtain
	\begin{equation}
		\frac{\partial}{\partial r}\left[\frac{1}{r}\nu'\exp\left(\frac{\nu-\lambda}{2}\right) \right] \leq 0.
	\end{equation}
	We integrate this equation from $R$ to $r<R$
	\begin{equation}
		\frac{\nu'}{r}\exp\left(\frac{\nu-\lambda}{2}\right)\geq\frac{2\nu'(R)}{R}\e^{-\frac{1}{2}\lambda(R)}\left.\frac{\partial}{\partial r}\e^{\frac{\nu}{2}}\right|_R
	\end{equation}
	and use the Schwarzschild solution at $r=R$ for $\e^\lambda$ and $\e^\nu$. 
	This is justified since we assumed $\rho(r)=0$ for $r>R$ and thus we need to recover the vacuum solution for a spherically symmetric object which is given by the Schwarzschild solution. 
	By continuity of the metric on every point of space, we can match 
	\begin{equation}
		\left.\e^{-\lambda(r)}\right|_R=\left[1-\frac{2M}{r}\right]_R=\left.\e^{\nu(r)}\right|_R
		\label{eq:3-Mass-Equ-MassLimitApprox1}
	\end{equation}
	and with the explicit solution for $\e^{-\lambda}$, we obtain
	\begin{equation}
		\left.\frac{2m(r)}{r}\right|_R = \frac{2M}{R}.
	\end{equation}
	When plugging this into equation~\eqref{eq:3-Mass-Equ-MassLimitApprox1}, the result is
	\begin{equation}
		\frac{\nu'}{2r}\exp\left(\frac{\nu-\lambda}{2}\right)\geq\frac{(1-2M/R)^{1/2}}{R}\left.\frac{\partial}{\partial r}\left(1-\frac{2M}{r}\right)^{1/2}\right|_{r=R} = \frac{M}{R^3}.
	\end{equation}
	Now we multiply by $r\exp(\lambda/2)$ and use the explicit solution for $\e^\lambda$
	\begin{equation}
		\frac{\partial}{\partial r}\left(\e^{\frac{\nu}{2}}\right) \geq \frac{M}{R^3}r\e^\frac{\lambda}{2} = \frac{M}{R^3}\left(r-2m(r)\right)
	\end{equation}
	and integrate again this time from $0$ to $R$
	\begin{equation}
		\e^{\nu(0)/2}\leq\left(1-\frac{2M}{R}\right)^{1/2}-\frac{M}{R^3}\int\limits_0^R\left[1-\frac{2m(r)}{r} \right]^{-1/2}r\diff r.
		\label{eq:3-Mass-Equ-MassLimitNuApprox}
	\end{equation}
	As we have already noted, the average density $m(r)/r^3$ decreases, meaning explicitly $m(r)/r^3\geq M/R^3$ and thus the integral with the previous equation can be written as
	\begin{align}
		\e^{\nu(0)/2}&\leq\left(1-\frac{2M}{R}\right)^{1/2}+\frac{1}{2}\left.\left[1-\frac{2Mr^2}{R^3} \right]^{1/2}\right|^{r=R}_{r=0}\\
		&= \frac{3}{2}\left(1-\frac{2M}{R}\right)^{1/2} - \frac{1}{2}.
	\end{align}
	The simple fact that $\e^{\nu(0)/2}>0$ then implies
	\begin{equation}
		\left(1-\frac{2M}{R}\right)^{1/2} > \frac{1}{3}
	\end{equation}
	which is equivalent to
	\begin{equation}
		M< \frac{4R}{9}.
	\end{equation}
\end{proof}\noindent
This shows that the mass of star has an upper limit under the assumptions $\rho\geq0$, $\partial_r\rho\leq0$ and $\rho(R)=0$ for some $R\geq0$.
In particular by restraining the total radius of the stellar object one can immediately find mass limits.
\newpage
\begin{section}{Numerical Solutions}
\begin{subsection}{Comparing \texorpdfstring{\acrNoHyperlink{\acs}{TOV}}{TOV} and \texorpdfstring{\acrNoHyperlink{\acs}{LE}}{LE} results with a polytropic \texorpdfstring{\acrNoHyperlink{\acs}{eos}}{EoS}}
\label{subsec:4-NumSol-Sec-Comp-TOV-LE}
This section presents numerical solutions of the \ac{TOV} equation
\begin{align}
	\frac{\partial m}{\partial r} &= 4\pi\rho r^2\label{eq:4-NumSol-Equ-TOVEqBasic1}\\
	\frac{\partial p}{\partial r} &=-\frac{m\rho}{r^2}\left(1+\frac{p}{\rho}\right)\left(\frac{4\pi r^3 p}{m}+1\right)\left(1-\frac{2m}{r}\right)^{-1}
	\label{eq:4-NumSol-Equ-TOVEqBasic2}
\end{align}
as derived previously in section~\ref{subsec:3-Mass-Sec-TOVDerivation}.
Results in this section are implemented and carried out by the author.
To obtain numerical solvability an \ac{eos} in the form $\rho(r,p)$ is supplied.
In Figure~\ref{fig:4-NumSol-Plt-TOVEqEasyEOS}, a plot of such a solution is presented.
The density $\rho$ is derived via the equation~\eqref{eq:4-NumSol-Equ-TOVEqBasic1} and the integration is done with a 4th order Runge-Kutta Method~\cite{rungeUeberNumerischeAufloesung1895, schlömilch1901zeitschrift, h.SimplifiedDerivationAnalysis2010}.
The integration is stopped once the pressure reaches values $p\leq0$.
The initial value of~\eqref{eq:4-NumSol-Equ-TOVEqBasic1} is $\partial_r m(r=0)=0$.
For equation~\eqref{eq:4-NumSol-Equ-TOVEqBasic1}, the initial value can be calculated when applying L'Hôpital's rule and combining them to obtain
\begin{alignat}{5}
	\lim\limits_{r\rightarrow0}\frac{m}{r} &= \lim\limits_{r\rightarrow0}\frac{\partial m}{\partial r} &&=\lim\limits_{r\rightarrow0}\frac{4\pi\rho r^2}{1} &&= 0\\
	\lim\limits_{r\rightarrow0}\frac{m}{r^2} &= \lim\limits_{r\rightarrow0}\frac{1}{2r}\frac{\partial m}{\partial r}  &&= \lim\limits_{r\rightarrow0}\frac{4\pi\rho r^2}{2r} &&= 0\\
	\lim\limits_{r\rightarrow0}\frac{m}{r^3} &= \lim\limits_{r\rightarrow0}\frac{1}{3r^2}\frac{\partial m}{\partial r} &&=\lim\limits_{r\rightarrow0}\frac{4\pi\rho r^2}{3r^2} &&= \frac{4\pi\rho_0}{3}\\
	\lim\limits_{r\rightarrow0}\frac{\partial p}{\partial r} &= 0
\end{alignat}
Explicit code can be found in~\cite{pleyerGithubRepositoryJonas2021}.
The Lane-Emden equation was obtained in section~\ref{subsec:3-Mass-Sec-LEDerivation} as the non-relativistic limit of the \ac{TOV} equation by neglecting terms of order $1/c^2$ and higher and setting $G=c=1$.
To obtain numerical results for the \ac{LE} equation as given in equation~\eqref{eq:3-Mass-Equ-Lane-Emden-Eq}, a substitution of the form $d\theta/d\xi=\chi$ was used.
\begin{equation}
	\begin{aligned}
		\frac{d\theta}{d\xi} &= \chi &\hspace{1cm} \left.\frac{d\theta}{d\xi}\right|_{\xi=0} &= 0\\
		\frac{d\chi}{d\xi} &= -\frac{2}{\xi}\chi-\theta^n & \left.\frac{d\chi}{d\xi}\right|_{\xi=0} &= -\theta^n
		\label{eq:4-NumSol-Equ-LE-Substitution}
	\end{aligned}
\end{equation}
The initial value for $d\chi/d\xi$ can be calculated with L'Hôpital's rule.
With the conversion factor $\kappa$ derived in section~\ref{subsec:3-Mass-Sec-LEDerivation}, \ac{TOV} and \ac{LE} results can be compared.
Figure~\ref{fig:4-NumSol-Plt-TOVEqEasyEOS} shows the solution of both equations for the parameters of Table~\ref{tab:4-NumSol-Tbl-TOVParameters}.
Additionally conversion equations to compare \ac{TOV} and \ac{LE} results are displayed.
Since the mass reads
\[
	m(r) = \int\limits_0^r 4\pi r'^2\rho(r')dr',
\]
we expect $\partial m/\partial r(R)=0$ if $p(R)=0$ when choosing a polytropic equation of state with $\gamma>0$.
The plot in Figure~\ref{fig:4-NumSol-Plt-TOVEqEasyEOS} shows this expected behaviour for the Lane Emden equation at $r\approx2.31$ and has the same behaviour for the \ac{TOV} results at $r\approx6.80$.\footnote{For the purpose of nicely displaying the calculated result, the plot only shows result up to $r=2.5$}
\begin{table}
	{\renewcommand{\arraystretch}{1.2}
	\centering
	\begin{tabular}{@{}llcll@{}}
		\toprule
		\multicolumn{2}{c}{\textbf{TOV}} & \phantom{b} &\multicolumn{2}{c}{\textbf{LE}}\\
		\cmidrule{1-2} \cmidrule{4-5}
		EOS & $\rho=Ap^{1/\gamma}$ && EOS & $p=K\rho^{\gamma}$\\
		$A$ & $2$ & & \\
		$\gamma=1+\frac{1}{n}$ & $4/3$ && $n=1/(\gamma-1)$ & $3$\\
		$p_0$ & $0.5$ && $\theta_0$ & $1$\\
		$m_0$ & $0$ && $d\theta_0$ & $0$\\
		$dr$ & $0.01$ && $d\xi=dr/\kappa$ & $\approx0.0298$\\
		\cmidrule{1-2} \cmidrule{4-5}
		$\rho_0=Ap_0^{1/\gamma}$ & $2(2)^{\frac{4}{3}}\approx1.1892$ && $\lambda=\rho_0$ & $2(2)^{\frac{4}{3}}\approx1.1892$\\
		&&& $K=A^{-1/\gamma}$ & $2^{-3/4}\approx0.5946$\\
		&&& $\kappa^2=((n+1)K\lambda^{1/n-1})/(4\pi)$ & $\approx0.1125$\\
		\bottomrule
	\end{tabular}
	}
	\caption[Numerical Parameters for \acrNoHyperlink{\acs}{TOV} and LE equation]{Numerical Parameters for \ac{TOV} and \ac{LE} equation}
	\label{tab:4-NumSol-Tbl-TOVParameters}
	\small
	To compare results, values $\rho_0,K,\kappa$ are calculated from ones supplied to the solving routine.
\end{table}%
\begin{figure}[H]
	{\centering
	\import{pictures/4-NumericalSolutions/}{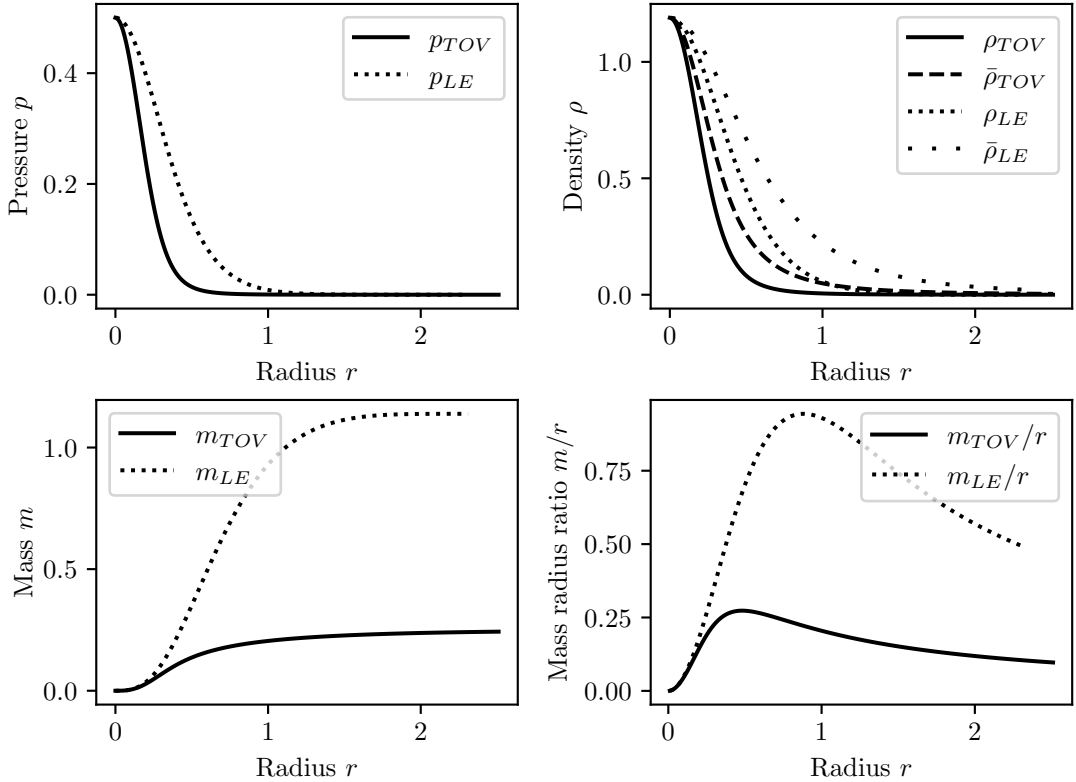}
	}%
	\caption[Comparison of the \acrNoHyperlink{\acs}{TOV} and \acrNoHyperlink{\acs}{LE} equation]{Comparison of the \ac{TOV} and \ac{LE} equation.}
	\label{fig:4-NumSol-Plt-TOVEqEasyEOS}
	\small
	The images show the plots for the parameters of Table~\ref{tab:4-NumSol-Tbl-TOVParameters}.
	The pressures and masses as presented in the first column are direct solutions of the \ac{TOV} or \ac{LE} equations.
	On the top right-hand side the densities calculated with the \ac{eos} and average densities $\bar{\rho}_i=(4\pi/3)^{-1}m_i/r^3$ are being compared.
\end{figure}
\end{subsection}
\begin{subsection}{Verifying the results}
\label{subsec:4-NumSol-Sec-Verifiying-the-results}
One can compare calculated \ac{LE} results with already known exact  solutions for certain exponents as given in table~\ref{fig:3-Mass-Tbl-LE-Exact-Results}.
The graphs of figure~\ref{fig:4-NumSol-Plt-ValidateLEResults} overall show good numerical agreement.
The integration stepsize for the left-hand side results of figure~\ref{fig:4-NumSol-Plt-ValidateLEResults} is $d\xi=0.03$.
Initial values at $r=0$ are identical since they were set to be.
The spike occurring afterwards can be explained by equations~\eqref{eq:4-NumSol-Equ-LE-Substitution}.
The initial value of all derivatives is identically $0$ which means no change in the values $\theta$ or $\chi$ for the initial step of the numerical integration.
On the other hand the exact results will show a change which explains the large discrepancy in the beginning.
Despite this behaviour we can see that the criterion for convergence is well fulfilled in the right-hand side of figure~\ref{fig:4-NumSol-Plt-ValidateLEResults}.
The plot displays nicely that by lowering the integration stepsize, the maximum difference in this interval decreases.
\begin{figure}[H]
	\centering
	\import{pictures/4-NumericalSolutions/}{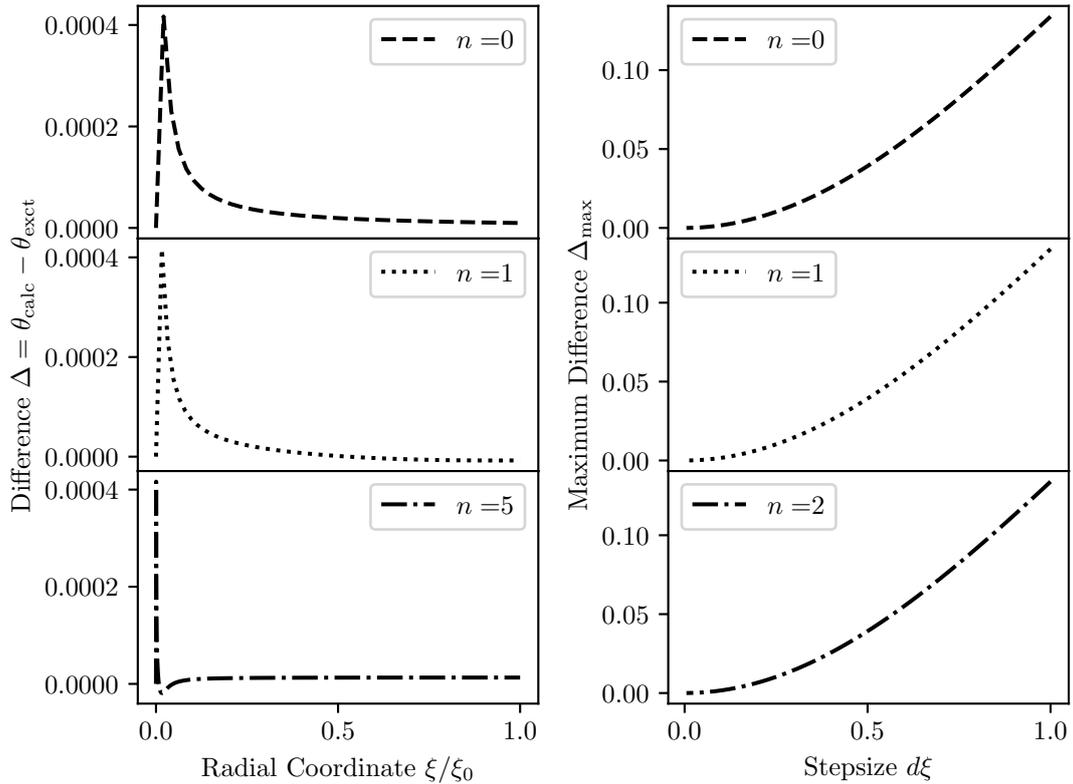}
	\caption{Validation of numerical LE results}
	The first two plots show absolute and relative difference of known exact and numerically calculated results while the third plot shows the behaviour of the maximum value of $\Delta$ as the stepsize decreases.
	\label{fig:4-NumSol-Plt-ValidateLEResults}
\end{figure}\noindent
%
\end{subsection}
\begin{subsection}{Relativistic \texorpdfstring{\acrNoHyperlink{\acs}{eos}}{EoS}}
\label{subsec:4-NumSol-Sec-RelEOS}
In the previous discussion, we relied on the \ac{eos} given in Table~\ref{tab:4-NumSol-Tbl-TOVParameters}.
This is a versatile assumption, but one could ask what would happen to a star in which the particles have no interaction but are close to relativistic speed.
The resulting \ac{eos} was calculated in the beginning (see equation~\eqref{eq:2-IntEner-FinalEOS}) although not written down explicitly.
Since explicit inversion of the given function is complicated, we rely on numerical methods for calculation.
Section~\ref{2-IntEner-SR-EOS-Derivation} showed that the function given in equation~\eqref{eq:2-IntEner-PressureAlpha}
\begin{equation}
	p(\alpha) = CNmc^2\frac{1}{K_2(\alpha)\alpha^2}\exp\left(-\alpha\frac{K_1(\alpha)+K_3(\alpha)}{2K_2(\alpha)}\right).
	\label{eq:4-NumSol-Pressure-Alpha-Dep}
\end{equation}
explicitly has a non-positive slope. 
When calculating an inverse function the values $C,N,m$ need to be chosen high enough in order to reach sufficiently high values for $p$.
By knowing the initial value $p_0$ and by the fact that the pressure of the \ac{TOV} solution~\eqref{eq:3-Mass-Equ-TOV-Eq-1} has a non-positive slope and with equation~\eqref{eq:2-IntEner-PressureAlpha-Limit}, it is sufficient to choose $Bp_0=CNmc^2>2/\e^2$.
On the other hand it is also important to choose high enough values for $\alpha$ for numerical inversion since otherwise functions can not be displayed if they reach small enough values for $p$.
The limit for $\alpha$ is chosen such that values as low as $0.001p_0$ can be evaluated.
Afterwards the function over the given interval is inverted by a polynomial fit and then used in equation~\eqref{eq:2-IntEner-DensityAlpha} to obtain the \ac{eos}.\\
Figure~\ref{fig:4-NumSol-Plt-RelEOS-TOV-Comparison} is an example which shows a clear difference between the polytropic and relativistic \ac{eos}.
The values used are $B=20$, $p_0=0.4$, $dr=0.001$, $n=1$.
The correct prefactor $A$ of the polytropic \ac{eos} is chosen after numerical calulation of the relativistic \ac{eos}.
It is then simply given by 
\begin{equation}
	A=\frac{\rho_\mathrm{rel}(p_0)}{p_0^{1/\gamma}}.
	\label{eq:4-NumSol-EOS-Factor-Explanation}
\end{equation}
This normalises both \acp{eos} to same inital pressures.
\begin{figure}[H]
	\centering
	\import{pictures/4-NumericalSolutions}{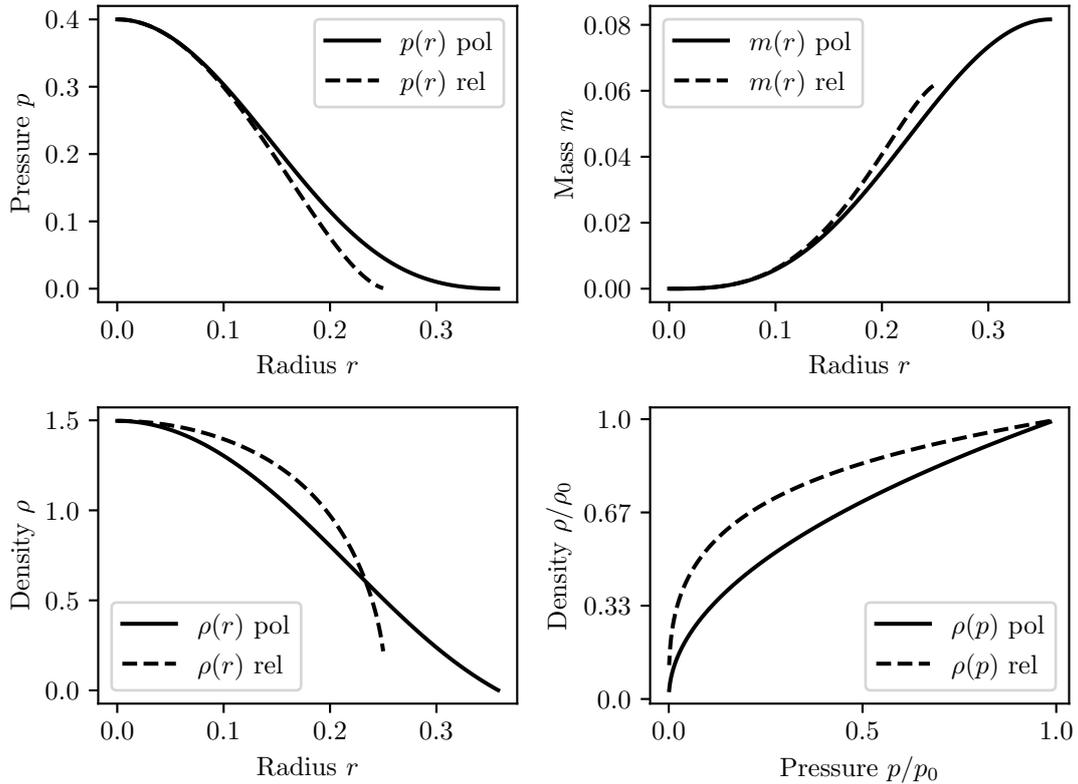}
	\caption[Comparison of TOV polytropic and relativistic EOS]{Comparison of TOV polytropic and relativistic \acs{eos}}
	\label{fig:4-NumSol-Plt-RelEOS-TOV-Comparison}
\end{figure}
\end{subsection}
\begin{subsection}{Zero Values of the \texorpdfstring{\acrNoHyperlink{\acs}{TOV}}{TOV} and \texorpdfstring{\acrNoHyperlink{\acs}{LE}}{LE} equation}
\label{subsec:4-NumSol-Sec-TOV-Exponents}
Having discussed individual solutions for the \ac{TOV} and \ac{LE} equation, this section now turns to analysing zero values of both equations.
We define a zero value as the first point in which the solution of our differential equation reaches value $0$, or more precisely where the pressure $p$ reaches zero.\footnote{Note that $p_0$ still denotes the initial value while $r_0$ and $\xi_0$ denote values where $p=0$.}
Even now the problem is not well posed since for different solving routines alternative transformations may be used.
This has to be kept in mind when comparing results. 
Apart from the previously defined exact solutions of Table~\ref{fig:3-Mass-Tbl-LE-Exact-Results} we will only use the parameter $r$ as defined by equation~\eqref{3-Mass-TOV-Eq} to display final results.
In order to solve the equations, a transformation in the form applied in section~\ref{subsec:3-Mass-Sec-LEDerivation} is used since it was hoped to increase the precision with which the zero value is determined.
This can be motivated by looking at plots for \ac{LE} solutions as in Figure~\ref{fig:3-Mass-Plt-LE-Exact-Results-Plots} and comparing them to \ac{TOV} results of Figure~\ref{fig:4-NumSol-Plt-TOVEqEasyEOS}.
With the \ac{eos} $\rho=Ap^{1/\gamma}$ and the new definition $\rho=\rho_0\theta^n$, the equations used are 
\begin{alignat}{3}
	\frac{\partial m}{\partial\xi} &= &&4\pi\kappa^3\rho\xi^2\\
	\frac{\partial\theta}{\partial\xi} &= -&&\frac{\left(1+K\rho_0^{1/n}\theta\right)\left(4\pi\xi^3\kappa^3 K\rho_0^{1+1/n}\theta^{n+1}+ m\right)}{\left((n+1)K\rho_0^{1/n}\kappa\xi^2\right)\left(1-\frac{2 m}{\kappa\xi}\right)}
	\label{eq:4-NumSol-Equ-TOV-Exponents-Transf-TOV}
\end{alignat}
where $\rho=\rho_0\theta^n$ and initial values $\partial m/\partial\xi=0$ and $\partial\theta/\partial\xi=0$ can again be calculated with L'Hospitals rule together with $\theta_0=1$ and $ m_0=0$.
When comparing equations~\eqref{eq:4-NumSol-Equ-TOV-Exponents-Transf-TOV} and~\eqref{3-Mass-TOV-Eq} one can immediately recognise the distinct terms.
When talking about zero values we stated earlier that transformations need be accounted when comparing results.
However, in this case since
\begin{equation}
	p=\frac{\rho_0^n\theta^n}{A^{1+1/n}}
	\label{eq:4-NumSol-Pressure-Theta-Relation}
\end{equation}
it is clear that zero values of $\theta$ and $p$ match under the same parameter $\xi$.
When transforming back by using $\kappa\xi=r$ as explained in section~\ref{subsec:3-Mass-Sec-LEDerivation} one should note that $\kappa$ depends on $n$ via
\begin{equation}
	4\pi\kappa^2=(n+1)K\rho_0^{1/n-1}
	\label{eq:4-NumSol-Conversion-Factor}
\end{equation}
and it will thus change values by more than constant scaling.
Figure~\ref{fig:4-NumSol-Plt-TOV-Exponents-Combo} shows zero values for solutions of the \ac{TOV} and \ac{LE} equations respectively with their dependence on $n$.
Since not only the parameters $A$ and $n$ in the \ac{eos} may alter results but also the initial value $p_0$ a family of curves was plotted.
The upper half of the plot displays a currently not explainable bump which is prominent in the \ac{TOV} $p_0=0.1$ case.
A similar but dampened behaviour can also be seen along the $p_0=1$ curve.
It is however still unclear if the \ac{TOV} solutions similarly to the \ac{LE} solutions blow up at the same combination of parameters.
In order to further analyse this property higher integration ranges for larger values of the polytropic index $n$ would be required.
This observation motivates a theorem that for every combination of parameters $A$ and $p_0$ there exists a value $n_0$ such that each solution of the TOV equation has no zero values for that particular combination of parameters $A,p_0$ if $n\geq n_0$.
The methods to prove this theorem will be discussed in the next section.\\
To obtain Figure~\ref{fig:4-NumSol-Plt-TOV-Exponents-Combo} the differential equations need to be solved multiple times for varying parameter combinations until the pressure reaches its zero value.
Numerical optimisations to speed up calculation times, and a short discussion on the methods used can be found in the appendix in section~\ref{subsec:99-App-Numerical-Optimisations}.
\begin{figure}[H]
	{\centering
	\import{pictures/4-NumericalSolutions/}{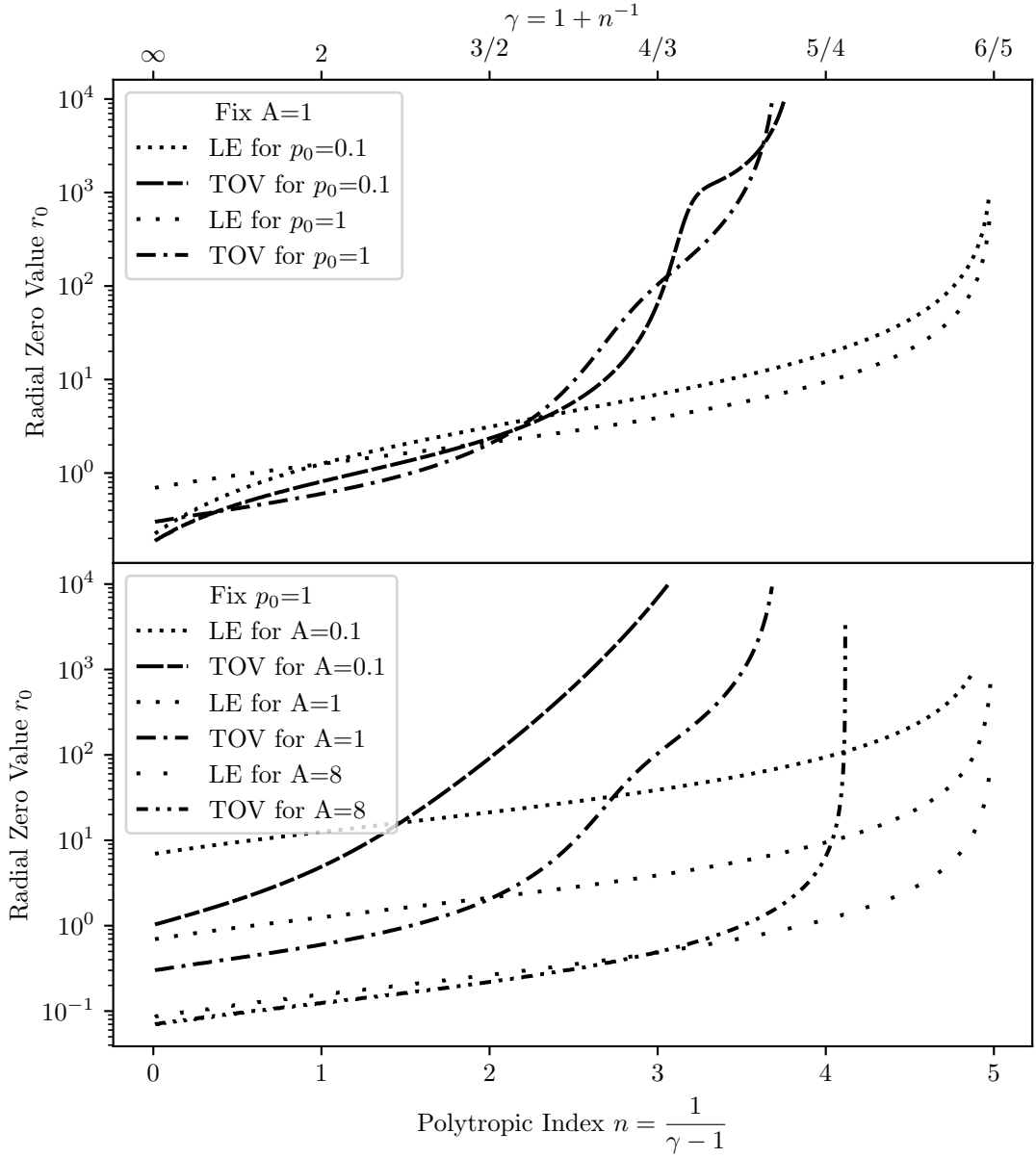}
	}
	\caption{Zero Values of TOV and LE equation}
	\label{fig:4-NumSol-Plt-TOV-Exponents-Combo}
	In the first plot results for $p_0=8$ have been omitted for better visual clarity.
	These behave similar to the $p_0=1$ results.
\end{figure}
\end{subsection}
\end{section}
\newpage
\section{To a General Zero Value Theorem}
\label{sec:5-zeroes}
The author will present proofs for uniqueness and existence for both equations.
Together with a result in~\cite{quittnerSuperlinearParabolicProblems2007}, a zero value theorem for the \ac{LE} equation is proved.
\subsection{\texorpdfstring{\acrNoHyperlink{\acs}{LE}}{LE} Equation}
\label{subsec:5-zeroes-le-equation}
Motivated by the preceding sections, the aim now is to prove the following theorem.
\begin{theorem}[\acrNoHyperlink{\acl}{LE} Finite Boundary]
	\label{5-Zeroes-Theo-Lane-EmdenFiniteBoundary}
	For every $0\leq n<5$, the generalised Lane-Emden equation
	\begin{equation}
		\frac{1}{\xi^{l-1}}\frac{d}{d\xi}\left(\xi^{l-1}\frac{d\theta}{d\xi}\right)+\theta^n=0
		\label{eq:5-Zeroes-Equ-LEeq}
	\end{equation}
	with $\theta_0=1$ and $l\geq2$ has a zero value for a finite $\xi_0$.
	To prove this theorem, we show the statements following in this section.
\end{theorem}
\begin{lemma}[\acrNoHyperlink{\acs}{LE} Local Existence]
	\label{5-Zeroes-Lem-Lane-Emden-Local-Existence}
	For every $n\geq0$, the Lane-Emden equation~\ref{eq:5-Zeroes-Equ-LEeq} with initial values defined as above has a unique solution in an $\epsilon$ Ball around $\xi=0$.
\end{lemma}\noindent
The procedure for proving existence is similar for both equations.
In a similar approach given by the author in~\cite[p.~50]{quittnerSuperlinearParabolicProblems2007}, we want to use the Schauder fixed-point theorem~\cite{schauderFixpunktsatzFunktionalraumen1930}, which is a variation the Banach fixed-point theorem to show that over a suitably chosen space of functions we obtain a unique solution to the differential equation.
The statement as formulated in~\cite{minazzoTheoremesPointFixe2007} reads
\begin{theorem}
	Every continuous function from a convex compact subset K of a Banach space to K itself has a fixed point.
\end{theorem}\noindent
In order to follow this approach, we examine some properties that solutions of the \ac{LE} differential equation will have.
\begin{lemma}
	\label{5-Zeroes-Lem-LE-Conditions}
	A positive solution $\theta$ of the \ac{LE} equation with $n\geq0$ on a suitably chosen interval $[0,\lambda]$ obeys the conditions
	\begin{align}
		\theta(0) 					&= \theta_0>0\label{5-Zeroes-Equ-LE-Conditions-Initial-1}\\
		\left.\frac{d\theta}{d\xi}\right|_{\xi=0} &= 0 \label{5-Zeroes-Equ-LE-Conditions-Initial-2}\\
		x\leq y 					&\Rightarrow \theta(y)\leq\theta(x)\label{5-Zeroes-Equ-LE-Conditions-3}\\
		\theta_0/2					&\leq \theta \label{5-Zeroes-Equ-LE-Conditions-4}\\
		|\theta(x)-\theta(y)|		&\leq L|x-y|\label{5-Zeroes-Equ-LE-Conditions-5}
	\end{align}
	where $L$ is independent of $\theta$.
	We will call $C^*\subseteq C^0$ the space of all continuous functions $f:[0,\lambda]\rightarrow\R$ that satisfy those conditions.
\end{lemma}
\begin{proof}
	The first and second initial conditions of equations~\eqref{5-Zeroes-Equ-LE-Conditions-Initial-1} and~\eqref{5-Zeroes-Equ-LE-Conditions-Initial-2} are true by assumption.
	The third condition~\eqref{5-Zeroes-Equ-LE-Conditions-3} can be obtained by integrating the \ac{LE} equation once
	\begin{equation}
		\frac{d\theta}{d\xi} = -\int\limits_0^\xi \frac{t^{l-1}}{\xi^{l-1}}\theta^n \diff t
		\label{eq:5-Zeroes-LE-Equation-Rewritten-Integration-1}
	\end{equation}
	where we not yet needed any of the above criteria except positivity of $\theta$.
	When integrating equation~\eqref{eq:5-Zeroes-LE-Equation-Rewritten-Integration-1} once more and using the initial value of condition~\eqref{5-Zeroes-Equ-LE-Conditions-Initial-1}, we obtain
	\begin{equation}
		\theta(\xi) = \theta_0 -\int\limits_0^\xi\int\limits_0^s \left(\frac{t}{s}\right)^{l-1}\theta^n(t) \diff t \diff s.
		\label{eq:5-Zeroes-LE-Equation-Rewritten-Integration-2}
	\end{equation}
	We combine this statement with the previously shown condition of equation~\eqref{5-Zeroes-Equ-LE-Conditions-3}.
	Then it is clear that
	\begin{equation}
		\int\limits_0^\xi\int\limits_0^s \left(\frac{t}{s}\right)^{l-1}\theta^n(t) \diff t \diff s \leq \theta_0^n\int\limits_0^\xi\int\limits_0^s \left(\frac{t}{s}\right)^{l-1} \diff t \diff s \leq \theta_0^n\frac{\xi^2}{2l}
		\label{eq:5-Zeroes-LE-Equation-Rewritten-Inequality}
	\end{equation}
	and thus $\xi\leq\lambda$ can be chosen suitably small to satisfy condition~\eqref{5-Zeroes-Equ-LE-Conditions-4}.
	For the last condition we derive from equation~\eqref{eq:5-Zeroes-LE-Equation-Rewritten-Integration-2} that
	\begin{equation}
		|\theta(x)-\theta(y)| = \left|\int\limits_x^y\int\limits_0^s \left(\frac{t}{s}\right)^{l-1}\theta^n(t) \diff t \diff s\right| \leq \theta^n_0\frac{|x^2-y^2|}{2l} \leq \theta^n_0\frac{\lambda}{l}|x-y|
		\label{eq:5-Zeroes-LE-Equation-Rewritten-Inequality-2}
	\end{equation}
\end{proof}\noindent
The now introduced function space $C^*$ contains all functions that are possible solutions for the \ac{LE} equation.
In order to use the Schauder fixed-point theorem, we have to show that $C^*$ is compact and convex which we will show in the next two lemmata.
\begin{lemma}
	The space $C^*$ defined in lemma~\ref{5-Zeroes-Lem-LE-Conditions} with the norm $||\cdot||_\infty$ is complete and convex.
	\label{5-Zeroes-Lem-LE-Convex-Complete}
\end{lemma}
\begin{proof}
	We begin by showing completeness.
	It is clear that conditions~\eqref{5-Zeroes-Equ-LE-Conditions-Initial-1} to~\eqref{5-Zeroes-Equ-LE-Conditions-4} are conserved.
	Generally speaking Lipschitz continuous functions need not be complete under $||\cdot||_\infty$.
	However, since we had chosen $L$ to be independent of $\theta$ for $||f_m-f||_\infty\rightarrow0$ we can write
	\begin{equation}
		|f(x)-f(y)|=|f(x)-f_m(x)+f_m(x)-f_m(y)+f_m(y)-f(y)|\leq 2\epsilon+L|x-y|
		\label{eq:5-Zeroes-LE-Equicontinuity-Proof}
	\end{equation}
	which proves the statement in the limit $\epsilon\rightarrow0$.
	Convexity for conditions~\eqref{5-Zeroes-Equ-LE-Conditions-Initial-1} to~\eqref{5-Zeroes-Equ-LE-Conditions-4} is also clear.
	Let $f,g\in C^*$ and $h\in[0,1]$.
	Then
	\begin{alignat}{3}
		&&&|hf(x)+(1-h)g(x)-hf(y)-(1-h)g(y)|\\
		&\leq h&&|f(x)-f(y)|+ (1-h)|g(x)-g(y)|\\
		&\leq L&&|x-y|
	\end{alignat}
	which shows convexity.
\end{proof}\noindent
\begin{lemma}
	The space $C^*$ with $||\cdot||_\infty$ is compact.
	\label{5-Zeroes-Lem-LE-Compact}
\end{lemma}
\begin{proof}
	To prove part of this statement we use the Arzela-Ascoli theorem.
	Since $C^*$ is a function space over a compact set $[0,\lambda]$ we only need to show that it is bounded and equicontinuous.
	For every function $f\in C^*$ we know by condition~\eqref{5-Zeroes-Equ-LE-Conditions-3} and~\eqref{5-Zeroes-Equ-LE-Conditions-4} that $\theta/2\leq f\leq\theta_0$ which shows boundedness.
	Since with condition~\eqref{5-Zeroes-Equ-LE-Conditions-5} every function $f\in C^*$ is especially Lipschitz, we can clearly see that equicontinuity hols.
\end{proof}
\begin{proof}[Proof \acrNoHyperlink{\acs}{LE} Local Existence]
	With these results it is now possible to prove our initial statement of theorem~\ref{5-Zeroes-Lem-Lane-Emden-Local-Existence}.
	It is left to prove that the operator $T:C^*\rightarrow C^*$ defined by
	\begin{equation}
		T(f)(x) = \theta_0 - \int\limits_0^x\int\limits_0^s \left(\frac{t}{s}\right)^{l-1}f^n(t)\diff t\diff s
		\label{eq:5-Zeroes-LE-Definition-Operator-On-C-Star}
	\end{equation}
	is continuous on $C^*$.
	With a quick integral approximation we obtain
	\begin{align}
		||T(f)-T(g)||_{\infty} &= ||\int\limits_0^x\int\limits_0^s \left(\frac{t}{s}\right)^{l-1}\left[f^n(t)-g^n(t)\right]\diff t\diff s\\
		&\leq \frac{x^2}{2l}||f^n-g^n||_{\infty}.
	\end{align}
	Since $z\mapsto z^n$ is continuous for $z\geq0$ and $n\geq0$ with $\mathrm{img}(f)\subseteq[\theta_0/2,\theta_0]$ for each $f\in C^*$ by condition~\eqref{5-Zeroes-Equ-LE-Conditions-4} the statement follows directly.
\end{proof}\noindent
Equation~\eqref{eq:5-Zeroes-LE-Equation-Rewritten-Integration-1} shows that non-negative solutions must have non-positive slope.
The following theorem makes use of that fact and of a different representation of the \ac{LE} equation to show that solutions can be extended until they reach zero.
\begin{lemma}
	Let $\theta$ be a \ac{LE} solution and $[0,\xi_\textrm{m}]$ be the maximum interval of existence.
	Then $\theta(\xi_\textrm{m})=0$ or $\xi_\textrm{m}=\infty$.
\end{lemma}
\begin{proof}
	It is clear that $\theta<0$ is a case that can never appear as long as we are interested in positive solutions.
	Suppose $\theta(\xi_\textrm{m})>0$ and $\xi_\textrm{m}<\infty$.
	Then there exists an interval $(\xi_\textrm{m}-\epsilon,\xi_\textrm{m}+\epsilon)$ and $(\theta(\xi_\textrm{m})-\delta,\theta(\xi_\textrm{m})+\delta)$ where the \ac{LE} differential equation has a solution.
	With equations~\ref{eq:4-NumSol-Equ-LE-Substitution} we can directly see that the Picard-Lindelöf theorem~\cite{lindelofApplicationMethodeApproximations1894} is applicable in this situation.
\end{proof}
\begin{proof}[Proof \acrNoHyperlink{\acl}{LE} Finite Boundary]
	To prove Theorem~\ref{5-Zeroes-Theo-Lane-EmdenFiniteBoundary} it is left to show that the \ac{LE} equation has no global solutions for $n<5$.
	In~\cite[p.~36]{quittnerSuperlinearParabolicProblems2007}, the author shows that for $n\leq5$, \ac{LE} solutions can not be global solutions.
	Furthermore it is explicitly shown that for $n\geq5$ global solutions and thus no zero value exist.
	A variation of the proof for this particular usecase can be found in appendix~\ref{subsec:99-App-NoGlobalLE}.
\end{proof}
\subsection{\texorpdfstring{\acrNoHyperlink{\acs}{TOV}}{TOV} Equation}
\label{subsec:5-Zeroes-TOV-Equ}
Motivated by the previous section about numerical solutions, we state the following hypothesis.
\begin{hypothesis}[\acrNoHyperlink{\acs}{TOV} Zero Value Hypothesis]
	\label{theo:5-Zeroes-TOV-Zero-Value-Hypothesis}
	Given the \ac{TOV} differential equation with polytropic \ac{eos} $\rho=Ap^{\frac{n}{n+1}}$ and $p_0,A>0$
	\begin{alignat}{3}
		\frac{\partial m}{\partial r} &= &&4\pi\rho r^2\\
		\frac{\partial p}{\partial r} &= -&&\frac{m\rho}{r^2}\left(1+\frac{p}{\rho}\right)\left(\frac{4\pi r^3 p}{m}+1\right)\left(1-\frac{2m}{r}\right)^{-1}
		\label{5-Zeroes-Equ-TOV-Equ}
	\end{alignat}
	and initial values $m_0=0,p_0\geq0$ there exists a $n_0\geq0$ such that all solutions with same parameters $A,p_0$ and smaller exponent $0\leq n<n_0$ have $p(r_0)$ for some $r_0>0$.
\end{hypothesis}
\begin{lemma}[\acrNoHyperlink{\acs}{TOV} Local existence]
	The \ac{TOV} equation~\ref{5-Zeroes-Equ-TOV-Equ} with initial values $m_0=0$ and $p_0\geq0$ and monotonously increasing continuous \ac{eos} $\rho:\R_{\geq0}\rightarrow\R_{\geq0}$ with $\rho(0)=0$ has a local solution for $r\in[0,\epsilon]$ with some $\epsilon>0$.
\end{lemma}
\begin{lemma}
	\label{5-Zeroes-Lem-TOV-Conditions}
	Let $(m,p):\R\rightarrow\R^2$ be a \ac{TOV} solution with initial values $m_0=0$ and $p_0>0$.
	Also let $\rho:\R_{\geq0}\rightarrow\R_{\geq0}$ be the polytropic \ac{eos}.
	Then for a small enough $\epsilon$ Interval $[0,\epsilon]$ the solution $(m,p)$ satisfies the following conditions.
	\begin{align}
		p(0)=p_0>0 &\mathrm{\ and\ } m(0)=m_0=0\label{5-Zeroes-Equ-TOV-Conditions-0}\\
		m(r) &\leq \frac{4\pi\rho_0 r^3}{3}\label{5-Zeroes-Equ-TOV-Conditions-1}\\
		\lim_{r\rightarrow0}\frac{m}{r^3}&=\frac{4\pi\rho_0}{3}\label{5-Zeroes-Equ-TOV-Conditions-2}\\
		p(r)&\geq\frac{p_0}{2}\label{5-Zeroes-Equ-TOV-Conditions-3}\\
		x\leq y&\Rightarrow m(x)\leq m(y)\label{5-Zeroes-Equ-TOV-Conditions-4}\\
		x\leq y&\Rightarrow p(x)\geq p(y)\label{5-Zeroes-Equ-TOV-Conditions-5}\\
		|m(x)-m(y)|+|p(x)-p(y)|&\leq L|x-y|\forall x,y\in[0,\epsilon]\forall(m,p)\in K\label{5-Zeroes-Equ-TOV-Conditions-6}
	\end{align}
\end{lemma}
\begin{proof}
	Conditions~\eqref{5-Zeroes-Equ-TOV-Conditions-0},\eqref{5-Zeroes-Equ-TOV-Conditions-4} and~\eqref{5-Zeroes-Equ-TOV-Conditions-5} are directly proven by looking at equation~\eqref{5-Zeroes-Equ-TOV-Equ}.
	To prove condition~\eqref{5-Zeroes-Equ-TOV-Conditions-1}, we inspect that with the initial value $m_0=0$ we can write
	\begin{equation}
		m = \int\limits_0^r 4\pi\rho r'^2 \diff r'\leq4\pi\rho_0\int\limits_0^r r'^2 \diff r'=\frac{4\pi\rho_0 r^3}{3}
		\label{eq:5-Zeroes-TOv-Conditions-Proof-Intermed-1}
	\end{equation}
	where in the second step we used that $\rho$ is monotonously increasing and positive solutions of the \ac{TOV} equation have non-positive slope in the pressure $p$ component.\\
	The third condition is obtained when applying L'Hosptial's Rule as was done in the beginning of section~\ref{subsec:4-NumSol-Sec-Comp-TOV-LE}.
	For the third condition we first inspect that condition~\ref{5-Zeroes-Equ-TOV-Conditions-1} lets us choose $(1-2m/r)^{-1}\leq 2$.
	\begin{align}
		p_0 - p &=\int\limits_0^r\frac{1}{r'^2}\left(p+\rho\right)\left(4\pi pr'^3+m\right)\left(1-\frac{2m}{r'}\right)^{-1}\diff r'\label{5-Zeroes-Equ-TOV-Conditions-Proof-1}\\
				&\leq2\left(p_0+\rho_0\right)\int\limits_0^r\frac{4\pi pr'^3+m}{r'^2}\diff r'\label{5-Zeroes-Equ-TOV-Conditions-Proof-2}\\
				&\leq2(p_0+\rho_0)\int\limits_0^r\left(4\pi p_0 r' + \frac{4\pi\rho_0 r'}{3}\right)\diff r'\label{5-Zeroes-Equ-TOV-Conditions-Proof-3}\\
				&=4\pi(p_0+\rho_0)\left(p_0+\frac{\rho_0}{3}\right)r^2\label{5-Zeroes-Equ-TOV-Conditions-Proof-4}
	\end{align}
	In the first step we used that $p$ is a \ac{TOV} solution while in the second step the aforementioned inequality and just as in the last step the non-positive slope of $p$ with the monotonicity of $\rho$ came into play.
	Next we used the condition~\ref{5-Zeroes-Equ-TOV-Conditions-1} to restrict $r\in[0,\epsilon]$ such that $(1-2m/r)^{-1}\leq2$ and integrated the result.
	Condition~\eqref{5-Zeroes-Equ-TOV-Conditions-6} can be shown when inspecting $m$ and $p$ individually.
	\begin{align}
		|m(x)-m(y)|&=\left|\int\limits_x^y 4\pi\rho r'^2 \diff r'\right|\leq4\pi\rho_0\epsilon^2|x-y|\\
		|p(x)-p(y)|&\leq8\epsilon\pi(p_0+\rho_0)\left(p_0+\frac{\rho_0}{3}\right)|x-y|
	\end{align}
	For the first equation we used the decreasing monotonous behaviour of $p$ and increasing behaviour of $\rho$.
	Afterwards, the just proven lines~\eqref{5-Zeroes-Equ-TOV-Conditions-Proof-1} to~\eqref{5-Zeroes-Equ-TOV-Conditions-Proof-4} and $x,y\leq\epsilon$ show the next statement.
	It is now clear that a small enough $\epsilon>0$ can be chosen such that all stated conditions are satisfied.
\end{proof}\noindent
We will now call $K$ the space of functions $f:[0,\epsilon]\rightarrow\R_{\geq0}^2$ that obey conditions~\eqref{5-Zeroes-Equ-TOV-Conditions-0} to~\eqref{5-Zeroes-Equ-TOV-Conditions-6}.
\begin{lemma}
	\label{5-Zeroes-Lem-K-Complete-Convex}
	The space $K$ with $||\cdot||_\infty$ is complete and convex.
\end{lemma}
\begin{proof}
	We need to show that conditions~\eqref{5-Zeroes-Equ-TOV-Conditions-0} to~\eqref{5-Zeroes-Equ-TOV-Conditions-6} are preserved under convexity and convergence in the $||\cdot||_\infty$ norm.
	We only need to show this for equation~\eqref{5-Zeroes-Equ-TOV-Conditions-6} which is true by the same argument as in the \ac{LE} case in lemma~\ref{5-Zeroes-Lem-LE-Convex-Complete}.
\end{proof}
\begin{lemma}
	The space $K$ is compact.
\end{lemma}
\begin{proof}
	We again use the Arzela-Ascoli Theorem.
	By condition~\eqref{5-Zeroes-Equ-TOV-Conditions-1} and $r\in[0,\epsilon]$ it is clear that all functions $f\in K$ are bounded.
	To show equicontinuity, we can apply the exact same argument as in lemma~\ref{5-Zeroes-Lem-LE-Compact} with condition~\eqref{5-Zeroes-Equ-TOV-Conditions-6} to obtain the desired result.
\end{proof}
\begin{proof}[Proof \acrNoHyperlink{\ac}{TOV} Existence]
	To show that local solutions of the \ac{TOV} equation exist, it remains to be proven that $U:K\rightarrow K$ defined by
	\begin{align}
		U_1((m,p))(r) &= \int\limits_0^r 4\pi\rho(p) r'^2 \diff r'
		\label{eq:5-Zeroes-TOV-Operator-Def-1}\\
		U_2((m,p))(r) &= -\int\limits_0^r\frac{p+\rho(p)}{r'^2}(4\pi\rho r'^3 + m)\left(1-\frac{2m}{r'}\right)^{-1}\diff r'
		\label{eq:5-Zeroes-TOV-Operator-Def-2}.
	\end{align}
	is continuous.
	For $(m,p),(u,q)\in K$ we respectively inspect the first and second component.
	By continuity of $\rho$, we have $|\rho(p)-\rho(q)|<\varepsilon$.
	Since by condition~\eqref{5-Zeroes-Equ-TOV-Conditions-0} and~\eqref{5-Zeroes-Equ-TOV-Conditions-3}, we know that we only need to consider $\rho$ on a compact interval $[p_0,p_0/2]$, it is clear that by choosing $||p-q||<\delta$ small enough the term in equation~\eqref{eq:5-Zeroes-TOV-Ineq-1} shrinks arbitrarily.
	\begin{equation}
		\left|\left|\int\limits_0^r 4\pi(\rho(p)-\rho(q))r'^2\diff r\right|\right|_\infty\leq \frac{4\pi r^3}{3}||\rho(p)-\rho(q)||_\infty.
		\label{eq:5-Zeroes-TOV-Ineq-1}
	\end{equation}
	For the second component, we make use of condition~\eqref{5-Zeroes-Equ-TOV-Conditions-1} to again restrict $r\in[0,\epsilon]$ such that $(1-2m/r)^{-1}\leq2$.
	\begin{align}
		\left|\left|U_2(m,p)-U_2(u,q)\right|\right|_\infty &\leq \int\limits_0^r(p_0+\rho_0)\left|\left|\left(4\pi r'(\rho(p))-\rho(q))+\frac{m-u}{r'^2}\right)\right|\right|_\infty\diff r'\\
		&\leq (p_0+\rho_0)M\left|\left|p-q\right|\right|_\infty + 8\pi\rho_0\epsilon(p_0+\rho_0)\left|\left|\frac{m-u}{r^2}\right|\right|_\infty
		\label{eq:5-Zeroes-TOV-Ineq-2}
	\end{align}
	By condition~\eqref{5-Zeroes-Equ-TOV-Conditions-2} we see that the term
	\begin{equation}
		\frac{m-u}{r^2}\rightarrow0
		\label{eq:5-Zeroes-TOV-Limit-UM-Div-R2}
	\end{equation}
	in the limit $r\rightarrow0$ and in general is bound by $8\pi\epsilon\rho_0/3$.
	This shows that for every $m,u$ there must exist a $\tau_0>0$ such that~\eqref{eq:5-Zeroes-TOV-Limit-UM-Div-R2} has its maximum value, or it is identically $0$.
	\begin{equation}
		\left|\left|\frac{m-u}{r^2}\right|\right|_{\infty} = \frac{1}{\tau_0^2}|m(\tau_0)-u(\tau_0)|\leq\frac{1}{\tau_0^2}||m-u||_{\infty}
		\label{eq:5-Zeroes-TOV-Last-Term}
	\end{equation}
	Let $m_k\rightarrow u$ be a series in $(K.||\cdot||_\infty)$.
	Without loss of generality, we can remove all elements of the series for which the maximum is identically $0$ and still have more than finite elements left.
	One can obtain another series $\tau_k$ by selecting the maximum point as done in equation~\eqref{eq:5-Zeroes-TOV-Last-Term}.
	This series is bound by assumption in $[0,\epsilon]$ and thus has a convergent subsequence.
	Consider $\tau_k\rightarrow\tau$.
	Then by definition we have
	\begin{equation}
		\left|\left|\frac{m_k-u}{r^2}\right|\right|_{\infty} = \frac{1}{\tau_k^2}\left|m_k(\tau_k)-u(\tau_k)\right| \leq \frac{1}{\tau_k^2} \left|\left|m_k-u\right|\right|_{\infty}.
		\label{eq:5-Zeroes-TOV-Last-Term-Series-2}
	\end{equation}
	If $\tau=0$, condition~\eqref{5-Zeroes-Equ-TOV-Conditions-2} implies that equation~\eqref{eq:5-Zeroes-TOV-Last-Term-Series-2} vanishes in the limit.
	In the case $\tau>0$, we can approximate equation~\eqref{eq:5-Zeroes-TOV-Limit-UM-Div-R2} by taking the smallest $\tau_{j}>0$ of the series and the result is proven again.
	This proves continuity of $T$ which in turn shows local existence of \ac{TOV} solutions.
\end{proof}%
\newpage
\section{Conclusion}
\label{sec:80-Outlook}
This thesis has dealt with the question to develop methods and describe zero values of the \ac{TOV} and \ac{LE} equation.
We demonstrated how to implicitly derive a fully special relativistic \ac{eos} in the example of a non-interacting gas (section~\ref{sec:2-Thermo}).
It was derived by using an explicit result for the partition function $\mathcal{Z}$ in equation~\ref{2-IntEner-PartFunc}, calculating the internal energy $\mathcal{U}$ and pressure $p$.
Afterwards an implicit relation between the energy density $\mathcal{U}/V$ and the pressure $p$ was found.
These results are also applicable in a wider context when generalising to special relativistic treatments of thermodynamics.
Future work could be carried forward in the context of interacting gases or by introducing quantum effects and varying constituents.\\
Further, the derived \ac{eos} and a standard polytropic \ac{eos} $\rho=Ap^{1/\gamma}$ were applied in numerical solutions of the \ac{TOV} equation.
Exact results in the non-relativistic \ac{LE} case (section~\ref{subsec:5-zeroes-le-equation}) show that this equation has finite zero values in the range $1<n<5$ of the polytropic index.
Numerical results of the relativistic \ac{TOV} case with a polytropic \ac{eos} demonstrate a similar behaviour.
We assess the question if hypothesis~\ref{theo:5-Zeroes-TOV-Zero-Value-Hypothesis} is provable by a similar procedure as in the \ac{LE} case.
We observe that its proof~\cite{quittnerSuperlinearParabolicProblems2007} (shortened version see section~\ref{subsec:99-App-NoGlobalLE}) heavily relies on pieces such as Bochner's identity or gradient estimates.
Although some of these statements directly survive when going to a curved space, it is unclear how the \ac{TOV} equation could be reformulated in terms of the Laplace operator on a spacelike hypersurface, in order to apply said arguments.
It's an affair that can be attributed to the non-linear first order nature of the Einstein equations~\ref{eq:01-Intr-Einstein-Equ}.\\
Another ansatz is to compare solutions to functions with known zero values.
Statements about global solutions such as Theorem XIII in~\cite[p. 99]{walterOrdinaryDifferentialEquations1998} could prove useful in this scenario but would require good approximations of the defect $(m,p)'-U(m,p)$ (see equations~\ref{eq:5-Zeroes-TOV-Operator-Def-1},\ref{eq:5-Zeroes-TOV-Operator-Def-2}) of the \ac{TOV} equation.\\
An additional question raised is how the behaviour of zero values changes if rotating or charged stars are being considered.
A stellar object in the Kerr metric~\cite{kerrGravitationalFieldSpinning1963} would be expected to yield smaller zero values compared to \ac{TOV} results due to a in general larger attractive force.
Since spherical symmetry in this case is not a valid assumption anymore, a new derivation is required which could however readily use the developed results and methods.
We expect similar behaviour in the charged case.\\
In the pursuit of classifying stellar model properties this thesis has improved on previous knowledge and left new and interesting questions unanswered.
The developed numerical methods are readily available and easy to use while the obtained results provide valuable background information to any theoretical astrophysicist.

\newpage
\appendix
\pagenumbering{roman}
\renewcommand{\thesection}{\Alph{section}}
\renewcommand{\thesubsection}{\Alph{subsection}}
\renewcommand\thetheorem{\thesubsection.\arabic{theorem}}

\section*{Appendix}
\addcontentsline{toc}{section}{Appendix}
\begin{subsection}{Known Exact Solutions of the \texorpdfstring{\acrNoHyperlink{\acs}{LE}}{LE} Equation}
\label{99-App-A-Exact-LE-Solutions}
This section relies on information in~\cite{weissteinLaneEmdenDifferentialEquation2020} and~\cite{chandrasekharChandrasekharAnIntroductionStudy1958}.
The \ac{LE} equation is
\begin{equation}
	\frac{1}{\xi^2}\frac{d}{d\xi}\left(\xi^2\frac{d\theta}{d\xi}\right)+\theta^n=0
	\label{eq:99-App-A-LE-Equation}
\end{equation}
which for $n=0$ transforms readily into
\begin{equation}
	\int\limits_0^\xi\frac{d}{d\xi}\left(\xi'^2\frac{d\theta}{d\xi}\right)d\xi' = -\int\limits_0^\xi\xi'^2 d\xi'.
	\label{eq:99-App-One-Integral-LE}
\end{equation}
Both sides can be evaluated directly and then simplified further.
\begin{align}
	\xi^2\frac{d\theta}{d\xi} &= -\frac{\xi^3}{3}\\
	\theta(\xi) &= \theta(0)-\frac{\xi^2}{6}
\end{align}
With the initial condition $\theta(0)=1$, we obtain the desired result.
For $n=1$, equation~\ref{eq:99-App-A-LE-Equation} transforms into
\begin{equation}
	\frac{d}{d\xi}\left(\xi^2\frac{d\theta}{d\xi}\right)+\xi^2\theta=0
	\label{eq:99-App-LE-Normal}
\end{equation}
we start by substituting $U=\theta x$ and thus obtain (also multiplying by $x^2$)
\begin{align}
	0 &= \frac{dU}{d\xi} + \xi\frac{\diff^2  U}{d\xi^2} - \frac{dU}{d\xi} + \xi U\\
	-U &= \frac{\diff^2  U}{d\xi^2}
\end{align}
and the last equation can be solved with a linear combination of $\cos(\xi)$ and $\sin(\xi)$.
Transforming back to $\theta$, we then have
\begin{equation}
	\theta(\xi) = A\frac{\sin(\xi)}{\xi} - B\frac{\cos(k\xi)}{\xi}.
	\label{eq:99-App-LE-Sol-N1-Ansatz}
\end{equation}
The need for a well defined limit at $\xi\rightarrow0$ implies that $B=0$ and thus since $\sin(z)/z\rightarrow1$ for $z\rightarrow0$, we have $A=\theta(0)=1$ and
\begin{equation}
	\theta(\xi) = \frac{\sin(\xi)}{\xi}.
	\label{eq:99-App-LE-Sol-N1-Result}
\end{equation}
For $n=5$, we start by making the two substitutions $x=1/\xi$ and $\theta=ax^\omega$ as done in~\cite[94\psqq]{chandrasekharChandrasekharAnIntroductionStudy1958} and then transforming the \ac{LE} equation into
\begin{alignat}{5}
	&x^4\frac{d\theta}{d\xi}&&+\theta^n&&=0\\
	&a\omega(\omega-1)x^{\omega+2}&&+a^n x^{n\omega} &&=0.
\end{alignat}
From this we see that $\omega+2=n\omega$ and $a\omega(\omega-1)=-a^n$ needs to be satisfied, since the equation needs to hold for all $x\in\R_{\geq0}$.
Rewriting these conditions, we obtain the singular solution for the \ac{LE} equation
\begin{equation}
	\theta(x) = \left(\frac{2(n-3)}{(n-1)^2}\right)^{1/(n-1)}x^{2/(n-1)}
	\label{eq:99-App-LE-Sol-N2}
\end{equation}
since with $x=1/\xi$, we have $\theta(\xi)\rightarrow\infty$ for $\xi\rightarrow0$.
Notice that the solution is only valid if $n\geq3$.
We can use this solution to the \ac{LE} equation and perturb it to make the more general ansatz
\begin{equation}
	\theta(x) = ax^\omega z(x).
	\label{eq:99-App-LE-Sol-Singular}
\end{equation}
If $n<3$, the factor $a$ has to be replaced with a more general one.
The singular solution is obtained when taking $z=1$.
Using another transformation $1/x=\xi=\exp(-t)$, we obtain
\begin{alignat}{5}
	&\frac{1}{\xi^2}\frac{d}{d\xi}\left(\xi^2\frac{d\theta}{d\xi}\right) &&+ \theta^n &&=0\\
	&x^4\frac{\diff^2 \theta}{dx^2} &&+ \theta^n &&=0\\
	ax^{\omega+2}\Bigl[&x^2\frac{\diff^2 z}{dx^2} + 2\omega x\frac{dz}{dx} + \omega(\omega-1)z\Bigr] &&+ a^n x^{n\omega}z^n &&=0\\
	&\frac{\diff^2 z}{dt^2} + (2\omega-1)\frac{dz}{dt}+\omega(\omega-1)z &&+ a^{n-1}z^n &&=0\\
	&\frac{\diff^2 z}{dt^2} + \frac{5-n}{n-1}\frac{dz}{dt} + 2\frac{3-n}{(n-1)^2}z &&+ 2\frac{(n-3)}{(n-1)^2}z^n &&=0.
\end{alignat}
For $n=5$, we obtain
\begin{equation}
	\frac{\diff^2 z}{dt^2}=\frac{1}{4}z(1-z^4).
	\label{eq:99-App-LE-Transform-N5}
\end{equation}
we multiply both sides with $dz/dt$ and integrate
\begin{equation}
	\frac{1}{2}\left(\frac{\diff^2 z}{dt^2}\right) = \frac{1}{8}z^2-\frac{1}{24}z^6+D
	\label{eq:99-App-A-LE-For-n-5-In-z-writing}
\end{equation}
where $D$ is the integration constant.
For $\xi\rightarrow0$, we expect $\theta\rightarrow\theta_0$ and thus $z=\theta_0\e^{-\omega t}(1/a+\mathcal{O}(\e^{-t}))$ as $t$ approaches $\infty$.
We immediately see that $dz/dt$ exhibits a similar behaviour and thus the integration constant $D$ has to vanish.\\
The right-hand side of equation~\ref{eq:99-App-A-LE-For-n-5-In-z-writing} cannot get negative since otherwise $z$ would take complex values which enables us to take the square root with a minus sign\footnote{This only determines the direction in which $t$ is defined, so it is arbitrary.} and integrate again
\begin{equation}
	\int\left(1-\frac{1}{3}z^4\right)^{-1/2}\frac{dz}{z}=-\frac{1}{2}\int dt.
	\label{eq:99-App-LE-N5-Intermediate-Eq}
\end{equation}
We change the integration by again substituting the variables $1/3z^4=\sin^2(\alpha)$ and calculate $dz/z=2\cos(\alpha)/\sin(\alpha)d\alpha$ as well as $-dt=d\alpha/\sin(\alpha)$.
Now we rewrite the integral
\begin{align}
	\int\frac{1}{\sqrt{1-\sin^2(\alpha)}}\frac{2\cos(\alpha)}{\sin(\alpha)}\frac{d\alpha}{} &= -\int dt\\
	\int\frac{2d\alpha}{\sin(\alpha)} &= -\int dt.
\end{align}
Evaluating those integrals leads us to
\begin{equation}
	\log(\tan(\alpha/2))+\log(1/C) = -t
	\label{eq:99-App-A-LE-For-n-5-Integral}
\end{equation}
where the integration constant has been chosen in advance to simplify the next expressions.
From here, we further manipulate equation~\ref{eq:99-App-A-LE-For-n-5-Integral} and combine it with our previous substitution $1/3z^4=\sin^2(\alpha)$ to obtain
\begin{equation}
	\frac{1}{3}z^4=\sin^2(\alpha)=\frac{4\tan^2(\alpha/2)}{\left(1+\tan^2(\alpha/2)\right)}^2
	\label{eq:99-App-LE-N5-Intermediate-Eq2}
\end{equation}
and with the solution for out integral before and plugging in the substitution from the beginning $\xi=\e^{-t}$, we have
\begin{equation}
	z=\pm\left(\frac{12C^2\xi^2}{1+C^2\xi^2}\right)^{1/4}.
	\label{eq:99-App-LE-N5-Intermediate-Eq3}
\end{equation}
For $\theta$, we need to have $\theta\rightarrow\theta0=1$ as $\xi\rightarrow0$.
This means that $C=1$ and with $\theta=ax^\omega z$, we obtain
\begin{equation}
	\theta = \frac{1}{\left(1+\frac{1}{3}\xi^2\right)^{1/2}}.
	\label{eq:99-App-LE-N5-LE-Result}
\end{equation}
We see that this equation has no zero value and tends to $0$ as $\xi\rightarrow\infty$.
\end{subsection}
\subsection{New Exact \texorpdfstring{\acrNoHyperlink{\acs}{LE}}{LE} Series Solution at Index 2}
\label{subsec:99-app-new-exact-le-sol-n2}
Although not the primary concern of this thesis, the author was able to derive another exact solution for the \ac{LE} equation at the exponent $n=2$ when
considering a series expansion of $\theta$ around the point $\xi=0$ with initial conditions
\begin{equation}
	\theta=\sum\limits_{m=0}^\infty a_m\xi^m \hspace{1cm} a_0=\left.\theta\right|_{\xi=0}=\theta_0
	\hspace{1cm} a_1=\left.\frac{d\theta}{d\xi}\right|_{\xi=0}=0
	\label{eq:99-App-LEN2-Series-Initial-Conditions}
\end{equation}
Since the series is absolut convergent, we can plug it into the \ac{LE} equation~\ref{eq:3-Mass-Equ-Lane-Emden-Eq} and use the Cauchy Product formula.
\begin{equation}
	\sum\limits_{m=2}^\infty m(m-1)a_m\xi^{m-2}+\sum\limits_{m=1}^\infty (2m)a_m\xi^{m-2} +
	\sum\limits_{m=0}^\infty\sum\limits_{k=0}^m a_{m-k}a_k\xi^m = 0
	\label{eq:99-App-LEN2-Series-PluggedIn}
\end{equation}
\begin{theorem}
	The odd coefficients $a_{2m+1}$ of this series expansion vanish.
\end{theorem}
\begin{proof}
	We rewrite the summations of equation~\ref{eq:99-App-LEN2-Series-PluggedIn} to start at the
	same index $m=0$ and combine them
	\begin{equation}
		\sum\limits_{m=0}^\infty\left((m+2)(m+1)a_{m+2}\xi^{m}+(2m+2)a_{m+1}\xi^{m-1} + \sum\limits_{k=0}^m a_{m-k}a_k\xi^m\right) = 0
		\label{eq:99-App-LE-New-Recursive-1}
	\end{equation}
	With equation~\ref{eq:99-App-LEN2-Series-Initial-Conditions}, we can start the summation
	in the middle one index higher and separate the term $\xi^m$.
	This equation has to be true
	inside the radius of convergence of the series and thus needs to vanish for ambiguous $\xi$.
	\begin{equation}
		(m+2)(m+1)a_{m+2}+2(m+2)a_{m+2}+\sum\limits_{k=0}^m a_{m-k}a_k = 0
		\label{eq:99-App-LE-New-Recursive-2}
	\end{equation}
	and upon further manipulation results in the recursive description for the coefficients
	of the series
	\begin{equation}
		a_{m+2} = -\frac{1}{(m+2)(m+3)}\sum\limits_{k=0}^m a_{m-k}a_k.
		\label{eq:99-App-LEN2-Recursive-Coefficients}
	\end{equation}
	We show the statement by induction.
	For $a_1$ it is already true.
	Let the statement be true for all odd values $2k+1\leq2m+1$.
	Writing down $a_{2m+3}$ gives us
	\begin{equation}
		a_{2m+3} = -\frac{1}{(2m+3)(2m+4)}\left(a_0 a_{2m+1}+a_1 a_{2m}+\dots+a_{2m}a_1+a_{2m+1}a_0\right).
		\label{eq:99-App-LE-New-Recursive-3}
	\end{equation}
	It is clear that in the summation odd and even coefficients get paired and will thus
	vanish completely.
\end{proof}
This proof shows that we can restrict ourselves to the subsequence $b_m=a_{2m}$ and the subseries with initial values given by
\begin{equation}
	\theta = \sum\limits_{m=0}^\infty b_m\xi^{2m} \hspace{1cm} b_{m+1} =
	-\frac{1}{(2m+2)(2m+3)}\sum\limits_{k=0}^m b_{m-k}b_k \hspace{1cm} b_0=\theta_0
	\label{eq:99-App-LEN2-bn-Definition}
\end{equation}
\begin{theorem}
	The series $\theta=\sum\limits_{m=0}^\infty b_m\xi^{2m}$ converges for $\xi\leq1$.
\end{theorem}
\begin{proof}
	We start by showing that $|b_{m+1}|\leq1/(4m+6)$.
	Using the triangle inequality on~\ref{eq:99-App-LEN2-bn-Definition}, we obtain
	\begin{equation}
		|b_{m+1}| \leq \frac{1}{(2m+2)(2m+3)}\sum\limits_{k=0}^m|b_{m-k}b_k|.
		\label{eq:99-App-LE-New-Recursive-bn}
	\end{equation}
	For $m=0$, we have $b_1=-1/6$ and thus $|b_1|\leq1/3$ and so the statement is true for $m=0$.
	If the statement holds for all $m<m+1$, we have
	\begin{equation}
		|b_{m+1}| \leq \frac{m}{(2m+2)(2m+3)} = \frac{1}{2}\frac{m}{(m+1)}\frac{1}{(2m+3)}
		\leq\frac{1}{2}\frac{1}{(2m+3)}
		\label{eq:99-App-LE-New-Bn-Inequality}
	\end{equation}
	which proves the first statement.
	To see that $b_m$ is a alternating sequence, we again inspect~\ref{eq:99-App-LEN2-bn-Definition}.
	The first element of the series is $b_0=1$ and the second is $b_1=-1/6$ so the initial statement is again correct.
	Let $b_m$ be alternating for all $m<2m+1$.
	Then equation~\ref{eq:99-App-LEN2-bn-Definition} for odd values reads
	\begin{equation}
		b_{2m+1} = -\frac{1}{(4 m+4)(4 m+5)}\left(b_{2m}b_0+b_{2m-1}b_1+\dots+b_0 b_{2m}\right)
		\label{eq:99-App-LE-Bn-Recursion}
	\end{equation}
	This shows that if all odd values $b_{2k+1}$ are negative and all even ones are positive,
	then $b_{2m+1}$ is negative too.
	The same holds true when the index $2m+2$ is even.
	By the Leibniz criterion the series converges if $\xi\leq1$.
\end{proof}
\begin{figure}[H]
	{\centering
	\import{pictures/5-MoreExactSolutions/}{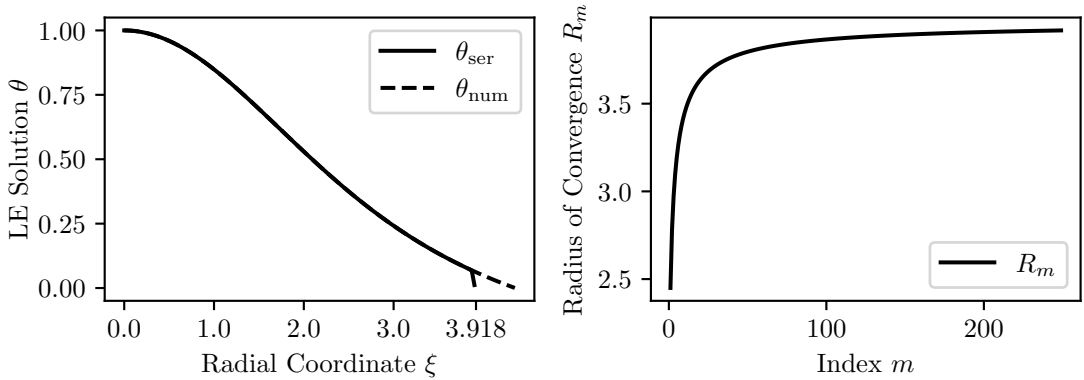}
	}%
	\caption[\acrNoHyperlink{\acs}{LE} Solution for $n=2$]{\ac{LE} Solution for $n=2$}
	\label{fig:99-App-Plt-LEN2-Plot}
	\small
	We see the power the series calculated with the coefficients explained above.
	The far right tick is the last calculated value for $R_m$ where plotting is stopped) and the sequence $R_m=(|a_m|)^{-1/m}$ which approaches the radius of convergence.
\end{figure}\noindent
These results can be used to calculate the power series of the \ac {LE} equation numerically.
Figure~\ref{fig:99-App-Plt-LEN2-Plot} shows those results.
Good agreement is found between the two solutions
as the difference $\Delta=\theta_{calc}-\theta_{ser}$ is not larger than $0.00015$.
The solutions differ for larger values around $\approx3.918$
The top right plot indicates that the radius of convergence might actually be larger than
where solutions deviate.
One important observation is that in this procedure, the coefficients were calculated up to $b_{250}\approx-4.7723\times10^{-296}$ which is in proximity to the floating point limit of python given by $\approx2.2251\times10^{-308}$ and indicates that numerical errors may be the reason the series starts to diverge.
Also in order to achieve sufficient numerical precision, we would need to calculate higher order terms of the series, which is supported by $(R_{250})^{500}\approx3.2159\times10^{296}$ and shows that these parts of the series still yield large contributions to the value of $\theta$ at this point.
\subsection{Exact \texorpdfstring{\acrNoHyperlink{\acs}{TOV}}{TOV} Solution in Absence of Mass}
\label{subsec:99-App-TOV-No-Mass}
In addition to the explicit \ac{LE} solution for $n=2$, the author derived a limiting case for solutions of the \ac{TOV} equation.
\begin{theorem}[TOV Exact Solution]
	The TOV equation~\ref{eq:4-NumSol-Equ-TOVEqBasic1} with a polytropic EOS $\rho=Ap^{1/\gamma}$ has a well defined limiting case where $A\rightarrow0$ with $m=0$ and
	\begin{equation}
		p = \frac{p_0}{2\pi rp_0+1}
		\label{eq:99-App-TOV-No-Mass-Sol}
	\end{equation}
\end{theorem}
\begin{proof}
	First we transform the TOV equation~\ref{3-Mass-TOV-Eq} using $p=y^a$ with $a>0$ and $m=Arv$ and together with the polytropic EOS $\rho=Ap^{1/\gamma}$ obtain
	\begin{align}
		\frac{\partial v}{\partial r} &= 4\pi ry^{a/\gamma}-Av\\
		ay^{a-1}\frac{\partial y}{\partial r} &= -\frac{vA^2 y^{a/\gamma}}{r}\left(1+\frac{y^a}{Ay^{a/\gamma}}\right)\left(\frac{4\pi r^2 y^a}{Av} +1\right)\frac{1}{1-2vA}
		\label{tmp-label-2}
	\end{align}
	Rearranging the second equation one obtains
	\begin{equation}
		\frac{\partial y}{\partial r} = -\frac{y^{a/\gamma-a+1}}{ar}\left(A+y^{a-a/\gamma}\right)\left(4\pi r^2 y^a +Av\right)\frac{1}{1-2vA}.
		\label{eq:99-App-TOV-No-Mass-Transformed-TOV-Eq}
	\end{equation}
	Using $\gamma=1+1/n$ with $n>0$, we see that this equation is continuous in every variable $(r,y,v)\in\R_{>0}\times\R_{\geq0}\times[0,1/2A)$.
	We restrict ourselves to a compact domain $r\in \overline{B_{\delta}(\tau)}$ and $v\in[0,1/2A-\epsilon]$ where $0<\epsilon<1/2A$.
	We choose $0<\delta$ and $0\leq\tau$ in such a way that $0<\epsilon<1/2A$ is satisfied for some $\epsilon$ which is possible since $m\rightarrow0$ for $r\rightarrow0$.
	To obtain Lipschitz continuity in $(y,v)$, all of the following conditions need to be fulfilled.
	\begin{equation}
		\frac{a}{\gamma}-a+1 \geq 1 \hspace{1cm} a-\frac{a}{\gamma} \geq 1 \hspace{1cm} a \geq 1
		\label{eq:99-App-TOV-No-Mass-Parameter-Condition}
	\end{equation}
	The second equation implies the first and the third.
	Thus we only need to choose $a\geq(1-1/\gamma)^{-1}$.
	With $\gamma=1+1/n$ we can rewrite this equation to
	\begin{equation}
		a\geq n+1
		\label{eq:99-App-TOV-No-Mass-Ineq-1}
	\end{equation}
	which can be easily satisfied.
	This now shows with extension of the Picard-Lindelöf Theorem that there exists a unique solution for given initial values $\tau, y(\tau), v(\tau)\in B_{\delta}(\tau)\times\R_{\geq0}\times[0,1/2A-\epsilon)$ that especially continuously depends on $A$.
	Without loss of generality we can choose $\delta$ small enough such that solutions are positive in $p$.
	By transforming back equation~\ref{tmp-label-2} to
	\begin{equation}
		\frac{\partial p}{\partial r} = -\frac{1}{r^2}\left(Ap^{1/\gamma}+p\right)\left(4\pi r^3 p+vA\right)\left(1-2vA\right)^{-1}
		\label{eq:99-App-TOV-No-Mass-TOV-Eq-Normal}
	\end{equation}
	and letting $A\rightarrow0$, we obtain
	\begin{equation}
		\frac{\partial p}{\partial r} = - 4\pi rp^2
		\label{eq:99-App-TOV-No-Mass-TOV-Eq-No-Mass}
	\end{equation}
	which is then solved by
	\begin{equation}
		p = \frac{\tilde{p}}{2\pi\tilde{p}(r^2-\tau^2)+1}
		\label{eq:99-App-TOV-No-Mass-Sol-1}
	\end{equation}
	where $\tilde{p}$ is the initial value at $r=\tau$.
	It is clear that this solution can be extended to $r\in[0,\infty)$ by choosing arbitrary small $\tau$.
	In the limit, the result equals the hypothesis.
\end{proof}
\subsection{Numerical Optimisations}
\label{subsec:99-App-Numerical-Optimisations}
In his success to solve the differential equations, numerical optimisations of the solving routine have been gathered by the author.
When solving the \ac{TOV} or \ac{LE} equation numerically calculations for the next step clearly depend on results from the before calculated one.
This simple fact which is integral to the concept of an \ac{ODE} means that numerics is single thread \footnote{Thread refers to a sequence of instructions given to a processing unit.} bound which already puts one at a systematic disadvantage.
This means any effort to parallelise the given task must focus on distributing the individual solving routines to respective threads.
With this first consideration at hand a server with high number of processing units may come into mind to then speed up the calculation.
However in the situation of section~\ref{subsec:4-NumSol-Sec-TOV-Exponents} this would only allow us to calculate more results in parallel but not speed up the individual solving routines.
And this is exactly the bottleneck in our case.
When looking at figure~\ref{fig:4-NumSol-Plt-TOV-Exponents-Combo} it is clear that higher values for $r_0$ are of interest where the curves reach the current maximum limit.
From this figure it is also clear that the effort will increase exponentially since we are bound to a single thread for individual solving routines.
Hence parallel computation of this problem would not help in this case especially considering that server processing units typically have lower single threaded performance than consumer chips.\\
Explicit parallelisation is achieved by distributing different values of $n$ to different threads.
More efficient distribution of those calculations can make a difference but effects discussed above and below have more impact.
Furthermore in order to fully understand how to improve for this distribution it is partly necessary to calculate the results which is then even less of an issue afterwards.\\
Another optimisation evolved around not computing unnecessary cases.
Consider again figure~\ref{fig:4-NumSol-Plt-TOV-Exponents-Combo}.
It is clear (however not proven until now) that for each combination $A,p_0$ at some value $n_0(A,p_0)$ values for $r_0$ will reach the numerical upper limit of the routine.
For other solving routines with $n\geq n_0$ the same will happen which makes it unnecessary to even begin to compute them since they will not be shown in the graph anyways.
This behaviour has been verified for a number of numerical combinations of parameters before beeing utilised in the final version.
For the solving routine this means explicitly that once there have been $2$ previous solving routines that failed in obtaining a value lower than $r_0$, the whole routine will not consider solving for any combinations of these particular parameters $A,p_0$ anymore.\\
Another obvious optimisation was to use a database to store results.
In this case MongoDB~\cite{dirolfPymongo11Python2021} was chosen for its good integration with the python language.
This pool of results allows one to compare the current solving routine with already calculated results.
Extra calculations can be permitted after a database query in the following cases:
\begin{itemize}
	\item If the current solving routine does have higher precision (in form of lower constant stepsize).
	\item If the previous solving routine reached the limit for $r_0$ and the current one would go further.
	\item If there exist multiple results for a combination of parameters $A,p_0,n$ \footnote{Then every result present is deleted and one new result will be calculated.}.
\end{itemize}
With these optimisations less results in total need to be calculated.
One still present disadvantage is however that allocation of parameters for solving routines is done beforehand and can not be changed in mid process.
This means if a thread started by the program finished and another is still running with multiple solving routines left, these cannot be given to another thread.
This is a possibility for future tweaks.\\
So far the previous discussion evolved around a fixed stepsize.
Since in this case the 4th order Runge Kutta solving method was implemented by hand, it is possible to use variable stepsizes for solving.
From figure~\ref{fig:4-NumSol-Plt-ValidateLEResults} we can see that the largest errors occur in the first step of integration.
This means in the beginning a small stepsize is desirable to retain accuracy while later it can be increased.
These settings were experimented with but since the author had already achieved numerous results with a constant stepsize were not committed to the final version.
It is however believed to yield a significant advantage if one would want to calculate higher values of $r_0$.\\
Furthermore, one could attempt to rewrite the code in a lower level programming language such as C to obtain a performance benefit.
Results from~\cite{jensenDrujensenFib2021} show a $676\%$ increase in performance for calculating Fibonacci numbers between the two languages when compiling code with PyPy.
It is however unclear to which degree this theoretical improvement would carry over to the code present.\\
The code is also designed in a way to check regularly for already present results in the database.
These calls could be forwarded to a different handler in order to possibly parallelise the problem and reduce disk\footnote{We refer to the term disk although the data was stored on a solid state drive.} latency issues.
\subsection{Nonexistence of Global \texorpdfstring{\acrNoHyperlink{\acs}{LE}}{LE} Solutions}
\label{subsec:99-App-NoGlobalLE}
We sketch the most important steps of the proof to theorem (8.1) part i) in~\cite[p.~36]{quittnerSuperlinearParabolicProblems2007} to show theorem~\ref{5-Zeroes-Theo-Lane-EmdenFiniteBoundary}.
\begin{lemma}
	\label{lem:99-App-NoGlobalLE-Lemm8-1}
	Non-negative solutions $u:\R^n\rightarrow\R$ of the differential equation
	\begin{equation}
		-\Delta u=u^p
		\label{eq:99-App-NoGlobalLE-LE-Equation}
	\end{equation}
	do not exist if $1<p<p_S$ where $p_S$ is the critical Sobolev-Exponent defined as
	\begin{equation}
		p_S=\begin{cases}
				\frac{n+2}{n-2} &\text{ if } n>2\\
				\infty &\text{ if } n\leq2
			\end{cases}
		\label{eq:99-App-NoGlobalLE-LE-Sobolev-Exponent}
	\end{equation}
	where $n$ is the dimension.
\end{lemma}
The following theorem can be used directly to prove the preceding one.
\begin{lemma}
	\label{lem:99-App-NoGlobalLE-Lemm8-6}
	Let $1<p<p_S$ and let $B_1$ be the unit ball in $\R^n$.
	There exists $r=r(n,p)>\max{(n(p-1)/2,p)}$ such that if $0<u\in C^2(B_1)$ is a solution of equation~\eqref{eq:99-App-NoGlobalLE-LE-Equation} in $B_1$, then
	\begin{equation}
		\int\limits_{|x|<1/2}u^r\leq C(n,p).
		\label{eq:99-App-NoGlobalLE-LE-Inequality}
	\end{equation}
\end{lemma}
To proof this theorem, a gradient estimate for a local solution of equation~\eqref{eq:99-App-NoGlobalLE-LE-Equation} was applied.
The following theorems are auxiliary in this task.
\begin{lemma}
	\label{lem:99-App-NoGlobalLE-Lemma1}
	Let $\Omega\subseteq\R^n$ and $0\leq\varphi$ be a test function in $C_c^\infty(\Omega)$ with $0\leq u\in C^2(\Omega)$.
	Then we fix $q\in R$ and define
	\begin{equation}
		I=\int\varphi u^{q-2}\left|\nabla u\right|^4, \hspace{1cm} J=\int\varphi u^{q-1}\left|\nabla u\right|^2\Delta u, \hspace{1cm} K=\int \varphi u^q(\Delta u)^2.
		\label{eq:99-App-NoGlobalLE-Def-IJK}
	\end{equation}
	For any $k\in\R$ with $k\neq -1$ we have
	\begin{equation}
		\alpha I + \beta J + \gamma K \leq \frac{1}{2}\int u^q\left|\nabla u\right|^2\Delta\varphi + \int u^q\left[\Delta u + (q-k)u^{-1}\left|\nabla u\right|^2\right]\nabla u\cdot\nabla\varphi
		\label{eq:99-App-NoGlobalLE-Ineq-IJK}
	\end{equation}
	where the parameters $\alpha, \beta, \gamma$ are defined as
	\begin{equation}
		\alpha=-\frac{n-1}{n}k^2 + (q-1)k-\frac{q(q-1)}{2}, \hspace{1cm} \beta=\frac{n+2}{n}k-\frac{3q}{2}, \hspace{1cm} \gamma = -\frac{n-1}{n}.
		\label{eq:99-App-NoGlobalLE-Def-Alph-Bet-Gam}
	\end{equation}
\end{lemma}
\begin{proof}
	The proof to this theorem relies on the Bochner identity
	\begin{equation}
		\frac{1}{2}\Delta |\nabla v|^2 = \nabla (\Delta v)\cdot \nabla v + |D^{2}v|^2
		\label{eq:99-App-NoGlobalLE-Bochner-Identity}
	\end{equation}
	where $|D^2 v|^2=\sum\limits_{i,j}(\partial_i\partial_j v)^2$.
	In the first step, equation~\eqref{eq:99-App-NoGlobalLE-Bochner-Identity} is multiplied by $\varphi v^m$ and integrated.
	Afterwards the Cauchy-Schwarz inequality applied on the inner product of two matrices $\left<A,B\right>=\mathrm{tr}(AB^*)$, results in
	\begin{equation}
		(\Delta v)^2 = (\mathrm{tr}(D^2 v))^2 \leq \mathrm{tr}\left[(D^2 v)(D^2 v)^*\right]\mathrm{tr}(I_n)=n|D^2 v|^2.
		\label{eq:99-App-NoGlobalLE-Step1-Result}
	\end{equation}
	When rearranging terms of equation~\eqref{eq:99-App-NoGlobalLE-Bochner-Identity}, the first result is obtained.
	\begin{alignat}{1}
		\frac{m(1-m)}{2} \int \varphi v^{m-2}|\nabla v|^{4}-\frac{3 m}{2} \int \varphi v^{m-1}|\nabla v|^{2} \Delta v-\frac{n-1}{n} \int \varphi v^{m}(\Delta v)^{2} \\
		\leq \frac{1}{2} \int v^{m}|\nabla v|^{2} \Delta \varphi+\int\left[v^{m} \Delta v+m v^{m-1}|\nabla v|^{2}\right] \nabla v \cdot \nabla \varphi
	\end{alignat}
	In the second part, we set $v=u^{k+1},m=(k+1)^{-1}(q-2k)$ and start by computing the respective terms of equation~\eqref{eq:99-App-NoGlobalLE-Step1-Result}.
	Combining them yields
	\begin{alignat}{1}
			{\left[\frac{m(1-m)}{2}(k+1)^{2}-\frac{3 m}{2} k(k+1)-\frac{n-1}{n} k^{2}\right] I+\left[-\frac{3 m}{2}(k+1)-2 k \frac{n-1}{n}\right] J} \\
			-\frac{n-1}{n} K \leq \frac{1}{2} \int u^{q}|\nabla u|^{2} \Delta \varphi+\int u^{q}\left[\Delta u+(k+m(k+1)) u^{-1}|\nabla u|^{2}\right] \nabla u \cdot \nabla \varphi,
	\end{alignat}
	which proves the lemma.
\end{proof}
\begin{lemma}
	\label{lem:99-App-NoGlobalLE-Lemm8-10}
	(i) Let $\Omega,\varphi,u$ be as in lemma~\ref{lem:99-App-NoGlobalLE-Lemma1} and $u$ additionally satisfies equation~\eqref{eq:99-App-NoGlobalLE-LE-Equation}.
	Fix $q,k\in\R$ with $q>-p,k\neq-1$ and define $I$ as in equation~\eqref{eq:99-App-NoGlobalLE-Def-IJK} with
	\begin{equation}
		K = \int\varphi u^{2p+q}.
		\label{eq:99-App-NoGlobalLE-K-NewDef}
	\end{equation}
	The it holds
	\begin{equation}
		\alpha I + \delta K \leq \frac{1}{2}\int u^q |\nabla u|^2\Delta \varphi + C\int\left[u^{p+q} +u^{q-1}|\nabla u|^2\right]|\nabla u\cdot \nabla\varphi|
		\label{eq:99-App-NoGlobalLE-lemma2-statement}
	\end{equation}
	where $C=C(n,p,q,k)>0$ and $\alpha$ as in equation~\eqref{eq:99-App-NoGlobalLE-Def-Alph-Bet-Gam} and
	\begin{equation}
		\delta=\frac{1}{p+q}\left(\frac{3q}{2}-\frac{n+2}{n}k\right)-\frac{n-1}{n}.
		\label{eq:99-App-NoGlobalLE-Def-Delt}
	\end{equation}
	(ii) Assume $1<p<p_S$.
	Then there exist $q,k\in\R$ with $q\neq -p,k\neq-1$ such that the constants $\alpha,\delta$ satisfy
	\begin{equation}
		\alpha,\delta>0\hspace{1cm}2p+q>n(p-1)/2.
		\label{eq:99-App-NoGlobalLE-Lem3-Ineq-Alph-Delt}
	\end{equation}
\end{lemma}
\begin{proof}
	(i) With $\Delta u+u^p=0$ and $J$ from equation~\eqref{eq:99-App-NoGlobalLE-Def-IJK} we obtain
	\begin{equation}
		(p+q)J = -\int\varphi u^{2p+q}+\int(\nabla\varphi\cdot\nabla u)u^{q+p}.
		\label{eq:99-App-NoGlobalLE-Lem3-Proof-i}
	\end{equation}
	(ii) In order to fulfill equation~\eqref{eq:99-App-NoGlobalLE-Lem3-Ineq-Alph-Delt} we rewrite the conditions.
	In the case for $\delta>0$ we obtain a polynomial in $q$ with discriminant
	\begin{equation}
		D=\frac{p^{2}+n(n+2) p-(n-1)^{2} p^{2}}{(n+2)^{2}}=\frac{n p[(n+2)-(n-2) p]}{(n+2)^{2}}>0.
		\label{eq:99-App-NoGlobalLE-Lem3-Proof-ii-Discriminant}
	\end{equation}
	Now we can simply choose
	\begin{equation}
		q=\frac{2p}{n+2}\hspace{1cm} k=k_0(q)^{-}=\left(-\frac{np}{n+2}\right)^- \hspace{1cm}\text{with} k\neq-1
		\label{eq:99-App-NoGlobalLE-Lem3-Proof-ii}
	\end{equation}
	where the superscript refers to the lower result of the quadratic polynomial in $q$ of equation~\eqref{eq:99-App-NoGlobalLE-Lem3-Proof-ii-Discriminant}.
\end{proof}
\begin{proof}[Proof of lemma~\ref{lem:99-App-NoGlobalLE-Lemm8-6}]
	With $k,q$ from lemma~\ref{lem:99-App-NoGlobalLE-Lemm8-10} (ii) and $\Omega=B_1$, we want to further develop estimates to the RHS of equation~\eqref{eq:99-App-NoGlobalLE-Lem3-Ineq-Alph-Delt}.
	Let $\xi$ be a testfunction on $\Omega$ with $\xi=1$ for $|x|\leq 1/2$ and $0\leq\xi\leq1$. Then we define $\theta=(3p+1+2q)/(2(2p+q))\in(0,1)$ and let $\varphi=\xi^m$ where $m=2/(1-\theta)$. Then
	\begin{equation}
		\left|\nabla\varphi\right|\leq C\xi^{m-1}\leq C\varphi^{\theta},\hspace{1cm}|\Delta\varphi|\leq C\xi^{m-2}\leq C\varphi^{\theta}.
		\label{eq:99-App-NoGlobalLE-Lem4-Proof-i}
	\end{equation}
	With Young's Inequality of the form
	\begin{equation}
		xyz\leq\epsilon x^a+\epsilon y^b+ C(\epsilon)z^c,\hspace{1cm}a^{-1}+b{-1}+c^{-1}=1
		\label{eq:99-App-NoGlobalLE-Lem4-Proof-ii}
	\end{equation}
	we get
	\begin{equation}
		\begin{gathered}
		\int u^{q}|\nabla u|^{2} \Delta \varphi=\int\left(\varphi^{1 / 2} u^{(q-2) / 2}|\nabla u|^{2}\right)\left(\varphi^{(q+2) / 2(2 p+q)} u^{(q+2) / 2}\right) \\
		\times\left(\varphi^{-(p+1+q) /(2 p+q)} \Delta \varphi\right) \leq \varepsilon \int \varphi u^{q-2}|\nabla u|^{4}+\varepsilon \int \varphi u^{2 p+q}+C(\varepsilon)
		\end{gathered}
		\label{eq:99-App-NoGlobalLE-Lem4-Proofiii-1}
	\end{equation}
	and
	$C \int u^{q-1}|\nabla u|^{2}|\nabla u \cdot \nabla \varphi| \leq \int\left(\varphi^{3 / 4} u^{3(q-2) / 4}|\nabla u|^{3}\right)\left(\varphi^{(q+2) / 4(2 p+q)} u^{(q+2) / 4}\right)$
	\begin{equation}
	\times\left(\varphi^{-(3 p+1+2 q) / 2(2 p+q)}|\nabla \varphi|\right) \leq \varepsilon \int \varphi u^{q-2}|\nabla u|^{4}+\varepsilon \int \varphi u^{q+2 p}+C(\varepsilon)
	\label{eq:99-App-NoGlobalLE-Lem4-Proofiii-2}
	\end{equation}
	Combining this with equation~\eqref{eq:99-App-NoGlobalLE-lemma2-statement}, we obtain
	\begin{equation}
		\alpha I+\delta K \leq C(n, p, q, k) \varepsilon(I+K)+C(\varepsilon).
		\label{eq:99-App-NoGlobalLE-Lem4-Proofiii-3}
	\end{equation}
	By choosing $\epsilon$ small enough, we obtain the desired result of lemma~\ref{lem:99-App-NoGlobalLE-Lemm8-6}.
\end{proof}

\pagebreak
\addcontentsline{toc}{section}{References}
\pagenumbering{alph}
\sloppy
\printbibliography

\includepdf{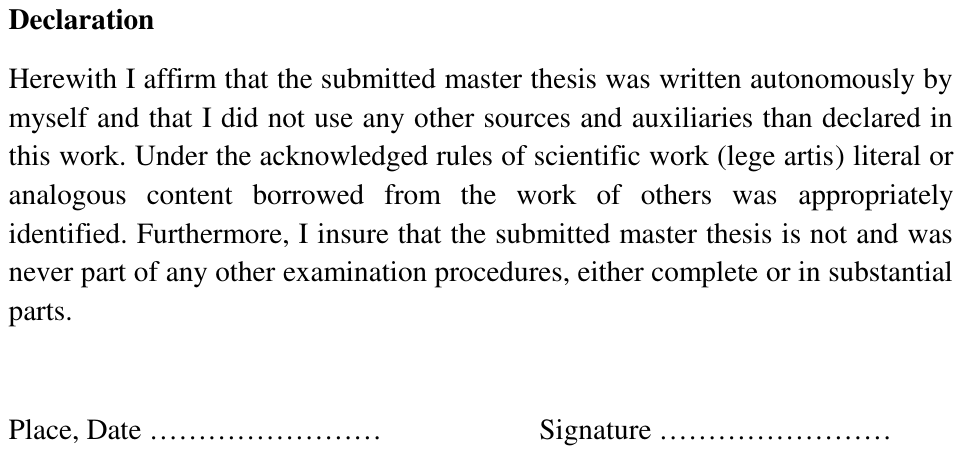}
\end{document}